\documentclass[preprint]{sigplanconf}

\newcommand{\pponly}[1]{}
\newcommand{\rronly}[1]{{#1}}

\newtheorem{example}{Example}

\newtheorem{proposition}{Proposition}

\newtheorem{theorem}{Theorem}
\newtheorem{definition}{Definition}
\newenvironment{proof}
{\par\noindent\textsc{Proof.}\;\sl}
{\hfill$\Box$\par}

\usepackage[latin1]{inputenc}
\usepackage[T1]{fontenc}
\usepackage{array,multirow,tabularx}
\usepackage{setspace}
\usepackage{latexsym,stmaryrd,amsmath,amssymb,mathtools}
\usepackage{mysymb}
\usepackage{epsfig,graphics,subfigure,wrapfig,alltt,verbatim}
\usepackage{pstricks,pstricks-add,pst-node,pst-plot}
\usepackage{color,myxcolor}
\usepackage{xspace,url}
\usepackage{paralist}
\usepackage{galois}
\usepackage{listings}
\usepackage{hyperref}
\usepackage{todonotes}

\lstdefinestyle{default}{%
  basicstyle={\small\ttfamily},%
  commentstyle={\slfamily},%
  keywordstyle=\bf,%
  columns=fullflexible,%
  keepspaces=true,%
  mathescape,%
  escapechar=\#,%
  escapeinside={/*}{*/},%
  escapebegin={ \color{black!70} -- },%
  escapeend={ },%
  xleftmargin=2em%
}
\hypersetup{
  colorlinks=true,
  breaklinks=true,
  urlcolor= blue,
  linkcolor= blue,
  bookmarksopen=true,
}

\newcommand\ie{\textit{i.e.}}
\newcommand\eg{\textit{e.g.}}

\newcommand\vx{\vec{x}}
\newcommand\vd{\vec{d}}
\newcommand\vb{\vec{b}}
\newcommand\vc{\vec{c}}
\newcommand\vv{\vec{v}}
\newcommand\reals{\mathbb{R}}
\newcommand\realsi{\bar{\mathbb{R}}}
\newcommand\true{\mathit{true}}
\newcommand\false{\mathit{false}}
\newcommand\denotation[1]{[\![ #1 ]\!]}
\newcommand\polyhed{\mathcal{CP}}
\newcommand\vr{\vec{r}}

\newcommand\vm{\vec{m}}

\newenvironment{spmatrix}
{\begin{small}\begin{pmatrix}}
{\end{pmatrix}\end{small}}

\newcommand{\vecb}[2]{\begin{pmatrix} #1 \\ #2 \end{pmatrix}}
\newcommand{\vect}[3]{\begin{pmatrix} #1 \\ #2 \\ #3 \end{pmatrix}}

\newcommand{\ceil}[1]{\lceil #1 \rceil}

\newcommand{\vceil}[1]{\left\lceil #1 \right\rceil}
\newcommand{\vfloor}[1]{\left\lfloor #1 \right\rfloor}

\newcommand{\diag}{\mathrm{Diag}}

\newcommand{\defeq}{\stackrel{\vartriangle}{=}}
\newcommand{\rpmatrix}[2]{
  \begin{pmatrix}
    #1 & - #2 \\
    #2 & #1
  \end{pmatrix}
}
\newcommand{\decomp}{\mathit{Frac}}

\newcommand\scrM{\mathcal{M}}
\newcommand\scrG{\mathcal{G}}

\title{Abstract Acceleration of General Linear Loops}

\authorinfo{Bertrand Jeannet}{INRIA}{bertrand.jeannet@inria.fr}
\authorinfo{Peter Schrammel}{University of Oxford}{peter.schrammel@cs.ox.ac.uk}
\authorinfo{Sriram Sankaranarayanan}{University of Colorado, Boulder}{srirams@colorado.edu}

\begin{document}

\maketitle


\begin{abstract}
  We present abstract acceleration techniques for computing loop
  invariants for numerical programs with linear assignments and
  conditionals. Whereas abstract interpretation techniques typically
  over-approximate the set of reachable states iteratively, abstract
  acceleration captures the effect of the loop with a single,
  non-iterative transfer function applied to the initial states at the
  loop head.  In contrast to previous acceleration
  techniques, our approach applies to any linear loop without
  restrictions. Its novelty lies in the use of the \emph{Jordan normal
    form} decomposition of the loop body to derive symbolic
  expressions for the entries of the matrix modeling the effect of $n
  \geq 0$ iterations of the loop. The entries of such a matrix depend
  on $n$ through complex polynomial, exponential and trigonometric
  functions. Therefore, we introduces an \emph{abstract
    domain for matrices} that captures the linear inequality relations
  between these complex expressions.  This results in an abstract
  matrix for describing the fixpoint semantics of the loop.

 Our approach integrates smoothly into standard abstract interpreters
 and can handle programs with nested loops and loops containing
 conditional branches.  
 We evaluate it over small but complex loops that are
 commonly found in control software, comparing it with other tools for
 computing linear loop invariants. The loops in our benchmarks
 typically exhibit polynomial, exponential and oscillatory behaviors
 that present challenges to existing approaches.  Our approach finds
 non-trivial invariants to prove useful bounds on the values of
 variables for such loops, clearly outperforming the existing
 approaches in terms of precision while exhibiting good performance.
\end{abstract}

\section{Introduction}
\label{sec:intro}

We present a simple yet effective way of inferring accurate loop
invariants of \emph{linear loops}, \ie~ loops containing linear
assignments and guards,
as exemplified by the programs shown in  Figs.~\ref{fig:parabola-code} and%
~\ref{fig:thermostat}. Such loops are particularly common in control
and digital signal processing software due to the presence of
components such as filters, integrators, iterative loops for equation
solving that compute square roots, cube roots and loops that
interpolate complex functions using splines. Static analysis of such
programs using standard abstract interpretation theory over polyhedral
abstract domains often incurs a significant loss of precision due to
the use of extrapolation techniques (widening) to force termination of
the analysis. However, widening is well-known to be too imprecise for
such loops. In fact, specialized domains such as ellipsoids and
arithmetic-geometric progressions were proposed to deal with two
frequently occurring patterns that are encountered in control loops
\cite{Feret04,Feret05b}.  These domains enable static analyzers for
control systems, \eg, \textsc{Astr\'ee}, to find the
strong loop invariants that can establish bounds on the variables or
the absence of run-time errors~\cite{CCFMMR09}.

In this paper, we present a promising alternative approach to such
loops by capturing the effect of a linear loop by means of a
so-called \emph{meta-transition} \cite{Boigelot96} that maps an
initial set of states to an invariant at the loop head. This process
is commonly termed \emph{acceleration}.
%
%
The idea of accelerations was first studied for communicating
finite-state machines \cite{Boigelot96} and counter automata
\cite{FL02}. Such accelerations can be either \emph{exact} (see
\cite{BFLP08} for a survey), or \emph{abstract}~\cite{GH06}.  Abstract
acceleration seeks to devise a transformer that maps initial sets of
states to the best correct over-approximation of
the invariant at the loop head for a given abstract domain, typically
the convex polyhedra domain.  Abstract acceleration enables static
analyzers to avoid widening for the innermost loops of the program by
replacing them by meta-transitions. As discussed in \cite{SJ12b} and
observed experimentally in \cite{SJ11a}, abstract acceleration
presents the following benefits w.r.t. widening:
\begin{compactenum}[(i)]
\item It is locally \emph{more precise} because it takes into
  account the loop body in the extrapolation it performs, whereas
  widening considers only sequences of invariants.
\item It performs \emph{more predictable} approximations because it is
  monotonic (unlike widening).
\item It makes the analysis \emph{more efficient} by speeding up
  convergence to a fixed point. For programs without nested loops, our
  acceleration renders the program loop-free.
\end{compactenum}
Apart from abstract interpretation, techniques such as
symbolic execution and bounded model-checking, that are
especially efficient over loop-free systems, can benefit from loop
acceleration.

\begin{figure}[t]
\psset{dotsize=3pt 0,arrows=->,nodesep=0.6ex}
\begin{tabular}{@{}l@{\hspace{1.5em}}r@{\hspace{0.1em}}l@{\hspace{6em}}l@{}}
&& \texttt{real x,y,z,t;} & \\
&& \multicolumn{2}{@{}l@{}}{\texttt{assume(-2<=x<=2 and -2<=y<=2 ~and -2<=z<=2);}} \\
&& \texttt{t := 0;} \\
\Rnode{lh1}{loop head} && \Rnode{lh}{$\bullet$} \\
&& \texttt{while \rnode{lg}{(x+y <= 30)}} & \rnode{lg1}{loop guard} \\
&& \multicolumn{2}{@{}l@{}}{\hspace*{1em} \texttt{\{ x := x+y; y := y+z; z := z+1; }} \\
&& \multicolumn{2}{@{}l@{}}{\hspace*{1em} \texttt{~~t := t+1; \}}} \\
\Rnode{le1}{loop exit}  && \Rnode{le}{$\bullet$}
\end{tabular}
  \ncline{lh1}{lh}
  \ncline{le1}{le}
  \ncline{lg1}{lg}
  \caption{Linear loop having a cubic behavior.}
  \label{fig:parabola-code}
\end{figure}



In this paper we present a novel approach to computing abstract
accelerations.  Our approach is \emph{non-iterative}, avoiding
widening. 
We focus on the linear transformation induced by a linear
loop body modeled by a square matrix $A$. We seek to
approximate the set of matrices $\{ I, A, A^2, \ldots \}$, which
represent the possible linear transformations that can be applied to
the initial state of the program to obtain the current state. This set
could be defined as fixed point equations on matrices and solved
iteratively on a suitable abstract domain for matrices. However, such
an approach does not avoid widening and suffers from efficiency
issues, because a matrix for a program with $n$ variables has $n^2$
entries, thus requiring a matrix abstract domain with $n^2$
different dimensions.


\paragraph{Contributions. } The overall contribution of this paper is
an abstract acceleration technique for computing the precise effect of
any linear loop on an input predicate. It relies on the
computation of the Jordan normal form of the square matrix~$A$ for the
loop body. Being based on abstract acceleration, it integrates
smoothly into an abstract interpretation-based analyzer and can be
exploited for the analysis of general programs with nested loops and
conditionals by transforming them into multiple loops around a 
program location.

The first technical contribution is an abstract acceleration
method for computing, non-iteratively, an approximation of the set
$\{ I, A, A^2, \ldots \}$ in an \emph{abstract domain for
  matrices}. It enables the analysis of any infinite, non-guarded
linear loop. The
main idea is to consider the Jordan normal form~$J$ of the
transformation matrix~$A$. Indeed, the particular
structure of the Jordan normal form~$J$  has two advantages:
\begin{compactenum}[(i)]
\item It results in closed-form expressions for the coefficients of
  $J^n$, on which asymptotic analysis techniques can be applied that
  remove the need for widening,
\item It reduces the number of different coefficients of $J^n$ to at
  most the dimension of the vector space (efficiency issue).
\end{compactenum}

This first contribution involves a conceptually simple but
technically involved
derivation that we omit in this paper and which can be found in
\pponly{the extended version~\cite{JSS13}}\rronly{the appendix}.

The second technical contribution addresses loops with guards that
are conjunctions of linear inequalities. We present an original
technique for bounding the number of loop iterations. Once again, we
utilize the Jordan normal form. These two techniques together make our
approach more powerful than ellipsoidal methods (\eg~\cite{RJGF12})
that are restricted to stable loops, because the guard is only weakly
taken into account.

We evaluate our approach by comparing efficiency and the
precision of the invariants produced with other invariant synthesis
approaches, including abstract interpreters and constraint-based
approaches.  The evaluation is carried out over a series of simple
loops, alone or inside outer loops (such as in
Fig.~\ref{fig:thermostat}), exhibiting behaviors such as polynomial,
stable and unstable exponentials, and inward spirals (damped
oscillators). We show the ability of our approach to \emph{discover
polyhedral invariants} that are sound over-approximations of
the reachable state space.  
%
For such systems, any inductive reasoning \emph{in} a linear domain as
performed by, \eg, standard abstract interpretation with Kleene
iteration and widening is often unable to find a linear invariant other than
$\true$. In contrast, our approach is shown to find useful bounds
for many of the program variables that appear in such loops.
To our knowledge, our method is the first one able to bound the
variables in the \textit{convoyCar} example of Sankaranarayanan et al. \cite{SSM04}.

\paragraph{Outline. } We introduce some basic notions in
\S\ref{sec:preliminaries}.  \S\ref{sec:overview} gives an overview of
the ideas of this paper.  \S\S\ref{sec:matrix} to \ref{sec:guard}
explain the contributions in detail.  \S\ref{sec:exp} summarizes our
experimental results and \S\ref{sec:rw} discusses related work
before \S\ref{sec:concl} concludes.

\section{Preliminaries}
\label{sec:preliminaries}

In this section, we recall the notions of linear assertions and
convex polyhedra, and we define the model of linear loops for
which we will propose acceleration methods.

\subsection{Linear assertions and convex polyhedra}
\label{sec:linear_assert_poly}

Let $x_1,\ldots,x_p$ be real-valued variables, collectively
forming a $p \times 1$ column vector $\vx$. A linear expression
is written as an inner product $\vc \cdot \vx$, wherein $\vc \in \reals^p$.
A linear inequality is of the
form $\vc \cdot \vx \leq d$ with $d \in \reals$.
A \emph{linear assertion} is a conjunction of linear
inequalities:
$
\varphi(\vx):\ \mathop{\bigwedge}_{i=1}^q \vc_i \cdot \vx  \leq
d_i
$.
The assertion $\varphi$ is succinctly written as $C \vx \leq \vec{d}$,
where $C$ is an $q \times p$ matrix whose $i^\text{th}$ row is
$\vc_i$. Likewise, $\vec{d}$ is an $q \times 1$ column vector whose
$j^\text{th}$ coefficient is $d_j$.
The linear assertion consisting of the single inequality $0 \leq 0$
represents the assertion $\true$ while the assertion $ 1 \leq 0$
represents the assertion $\false$.

Given a linear assertion $\varphi$, the set $\denotation{\varphi} = \{
\vx \in \reals^p\ |\ \varphi(\vx) \} $ is a convex polyhedron. The set
of all convex polyhedra contained in $\reals^p$ is denoted by
$\polyhed(\reals^p)$.  We recall that a convex polyhedron $P \in
\polyhed(\reals^p)$ can be represented in two ways:
\begin{enumerate}[(a)]
\item The \emph{constraint representation} $C \vx \leq \vd$ with matrix
$C$ and vector $\vd$.
\item The \emph{generator representation} with  a set of vertices
  $V= \{ \vv_1, \ldots, \vv_k \}$ and rays $R =
  \{\vr_1,\ldots,\vr_l\}$, wherein $\vx \in P$ iff
  $$
  \vx = \sum_{i=1}^k \lambda_i \vv_i  + \sum_{j=1}^l \mu_j
  \vr_j \text{ with } \lambda_i,\mu_j \geq 0 \text{ and }
  \textstyle\sum_i \lambda_i = 1
  \vspace*{-0.4ex}
  $$
\end{enumerate}

\subsection{Linear loops}
\label{sec:linearloop}

We consider \emph{linear loops} consisting of a
while loop, the body of which is a set of assignments without
tests and the condition is a linear assertion.
\begin{definition}[Linear loop]
A linear loop $(G,\vec{h}, A,\vb)$ is a program fragment of the form
\vspace*{-0.4ex}
$$
\mathsf{while}(G \vx \leq \vec{h})\ \vx := A \vx + \vec{b};
\vspace*{-0.4ex}
$$
where $\varphi:\ G\vec{x} \leq \vec{h}$ is a linear assertion over
the state variables
$\vx$ representing the loop condition and $(A,\vb)$ is the linear
transformation associated with the loop body.

\end{definition}



\begin{figure}[t]
  \begin{lstlisting}
real t,te,time;
assume(te=14 and 16<=t and t<=17);
while true {
  time := 0; /* timer measuring  duration in each mode */
  while (t<=22) { /* heating mode */
    t := 15/16*t-1/16*te+1; time++;
  }
  time := 0;
  while (t>=18){ /* cooling mode */
    t := 15/16*t-1/16*te; time++;
  }
}\end{lstlisting}
  \caption{A thermostat system, composed of two simple loops
    inside a outer loop.}
  \label{fig:thermostat}
\end{figure}

Figure~\ref{fig:parabola-code} shows an example of a linear
loop with a guard that computes $y=x(x+1)/2$ by the successive difference
method.  We give another example below.
\begin{example}[Thermostat]\label{Ex:iir-filter-example}
Figure~\ref{fig:thermostat} models the operation of a thermostat that
switches between the heating and cooling modes over time. The
variables $\mathtt{t}, \mathtt{te}$ model the room and outside
temperatures, respectively. We wish to show that the value of $t$
remains within some bounds that are close to the switch points $18,
22$ units.
\end{example}

Any linear loop $(G,\vec{h},A,\vb)$ can be \emph{homogenized} by
introducing a new variable $\xi$ that is a place holder for the
constant $1$ to a loop of the form
\vspace*{-0.4ex}
$$
\begin{array}{l}
  \mathsf{while}
  \left(
    \begin{pmatrix} G & \vec{h} \end{pmatrix}
    \begin{pmatrix} \vx \\ \xi \end{pmatrix}
    \leq \vec{0}
\right)\ \left\{
  \ \ \  \begin{pmatrix} \vx \\ \xi \end{pmatrix} := \begin{pmatrix} A  & \vb \\
    0 & 1 \end{pmatrix} \begin{pmatrix} \vx \\ \xi \end{pmatrix}; \right\}
\end{array}
$$ Henceforth, we will use the notation $(G \rightarrow A)$ to denote
the \emph{homogenized linear loop} $\mathsf{while}\ (G\vx\leq
\vec{0})\{\ \vx' := A\vx;\ \}$.

\begin{definition}[Semantic function]
The semantic function of a linear loop $(G \rightarrow A)$ over sets
of states is the functional
\vspace*{-0.4ex}
$$
(G \rightarrow A)(X) \defeq
A(X\cap\denotation{G\vx\leq 0}) \quad,\quad X\subseteq \mathbb{R}^p
\vspace*{-0.4ex}
$$
where $A(Y)$ denotes the image of a set $Y$ by the
transformation $A$. 
\end{definition}

\subsection{Convex and template polyhedra abstract domains}
\label{sec:polyabsdom}

The set of convex polyhedra $\polyhed(\reals^p)$ ordered by inclusion
is a lattice with the greatest lower bound $\sqcap$ being the set
intersection and the least upper bound $\sqcup$ being the convex hull.
The definition of the domain includes an abstraction function
$\alpha$ that maps sets of states to a polyhedral abstraction and a corresponding
concretization function $\gamma$.
We refer the reader to the original work of Cousot and Halbwachs for
a complete description~\cite{Cousot+Halbwachs/78/Automatic}.

It is well-known that the abstract domain operations such as
join and transfer function across non-invertible assignments
are computationally expensive. As a result, many \emph{weakly-relational}
domains such as octagons and templates have been proposed~\cite{SSM05,Mine/01/Octagon}.
Given a matrix $T\in\reals^{q\times p}$ of $q$ linear
expressions,  $\polyhed_T(\reals^p)\subsetneq \polyhed(\reals^p)$ denotes the set of
\emph{template polyhedra} on $T$:
\vspace*{-0.4ex}
$$
\polyhed_T(\reals^n) =
\Bigl\{ P\in\polyhed(\reals^p) \;|\; \exists \vec{u}\in \realsi^q :
P = \{ \vec{x} \;|\; T\vec{x} \leq \vec{u} \} \Bigr\}
\vspace*{-0.4ex}
$$
where $\realsi$ denotes $\reals\cup\{\infty\}$. A template
polyhedron will be denoted by $(T,\vec{u})$. If $T$ is is fixed,
it is uniquely defined by the vector $\vec{u}$.
$\polyhed_T(\reals^p)$ ordered by inclusion is a complete lattice.
The abstraction  $\alpha_T$ and concretization $\gamma_T$
are defined elsewhere~\cite{SSM05}.
%

\section{Overview}
\label{sec:overview}

This section provides a general overview of the ideas in this paper, starting
with abstract acceleration techniques.

\paragraph{Abstract Acceleration}
\label{sec:abstractaccel}
\label{sec:absacc:principle}
Given a set of initial states
$X_0$ and a loop with the semantic function $\tau$,
the smallest loop invariant $X$ containing $X_0$ can be formally written as
\vspace*{-0.4ex}
$$
X = \tau^*(X_0) \defeq \bigcup_{n\geq 0} \tau^n(X_0)
\vspace*{-0.4ex}
$$

Abstract acceleration seeks an ``optimal'' approximation of $\tau^*$
in a given abstract domain with abstraction
function $\alpha$~\cite{GH06}.
Whereas the standard abstract interpretation approach seeks
to solve the  fix point equation
$
Y' = \alpha(X_0)\sqcup
\alpha\bigl(\tau(Y')\bigr)
$
by iteratively computing
\vspace*{-0.0ex}
\begin{equation} \label{eq:absfixpoint}
Y=(\alpha\circ\tau)^*(\alpha(X_0))
\vspace*{-0.4ex}
\end{equation}
the abstract acceleration approach uses $\tau^*$ to compute
\vspace*{-0.4ex}
\begin{equation} \label{eq:absacc}
Z = \alpha\circ\tau^*(\alpha(X_0))
\vspace*{-0.4ex}
\end{equation}
Classically, Eqn.~(\ref{eq:absfixpoint}) is known as the minimal fixed
point (MFP) solution of the reachability problem whereas
Eqn.~(\ref{eq:absacc}) is called the Merge-Over-All-Paths (MOP)
solution. The latter is known to yield more precise results \cite{KU77}.

The technical challenge of abstract acceleration is thus to obtain
a closed-form approximation of $\alpha\circ\tau^*$ that avoids
both inductive reasoning in the abstract domain and the use of
widening.

\begin{figure*}[t]
  \centering

  \hfill
  \parbox{0.19\linewidth}{
    \centering
    \psset{unit=1.3em,griddots=5,subgriddiv=1,gridlabels=0pt,labelsep=0.2pt}
    \begin{pspicture}(0,-1)(5,5.5)
      \pspolygon[fillstyle=solid,fillcolor=LightGrey,linestyle=none]
      (5,0)(1,0)(1,0.75)(5,4.75)(5,4.75)
      \psdots(1,0)(1,0.75)
      \psline[arrowsize=3pt 4,arrowinset=0,arrows=->](1,0)(5,0)
      \psline[arrowsize=3pt 4,arrowinset=0,arrows=->](3,2.75)(5,4.75)
      \pspolygon[fillstyle=solid,fillcolor=darkgray]
      (1,0)(1,0.75)(3.25,3)(3.375,3)(3.375,2.375)
      \psgrid(0,0)(0,0)(5,5)
      \psaxes[arrows=->](0,0)(0,0)(5,5)
      \rput(5.8,0){{\small $m_1$}}
      \rput(0,5.35){{\small $m_2$}}
    \end{pspicture}

    \caption{
      Octagons defined by $\varphi_\scrM$ in
      Example~\ref{ex:running1} (light gray),
      and by $\varphi_{\scrM'}$ in Ex.~\ref{ex:running4} (dark
      gray).
    }
    \label{fig:running1}
  }\hfill
  \parbox{0.35\linewidth}{
    \centering
    \psset{unit=1.1em,griddots=5,subgriddiv=1,gridlabels=0pt,Dx=2,labelsep=0.2pt}
    \vspace{1ex}
    \begin{pspicture}(0,-1)(17,7.5)
      \pspolygon[fillstyle=solid,fillcolor=LightGrey,linestyle=none](1,0)(1,2)(1.5,3.75)(6.25,7)(17,7)(17,0)
      \psline(17,0)(1,0)(1,2)(1.5,3.75)(6.25,7)
      \psdots(1,0)(1,2)(1.5,3.75)
      \psline[arrowsize=3pt 4,arrowinset=0,arrows=->](1.5,3.75)(3,4.75)
      \psline[arrowsize=3pt 4,arrowinset=0,arrows=->](3,0)(5,0)
      \psframe[fillstyle=solid,fillcolor=darkgray](1,0)(3,2)
      \psline[linewidth=0.5pt,linestyle=dashed](1,2)(1.5,3)(2.25,4)(3.375,5)(5.06,6)(7.59,7)
      \psline[linewidth=0.5pt,linestyle=dashed](2,1)(3,2)(4.5,3)(6.75,4)(10.125,5)(15.1875,6)(17,6.2)
      \psline[linewidth=0.5pt,linestyle=dashed](3,0)(4.5,1)(6.75,2)(10.125,3)(15.1875,4)(17,4.2)
      \psdots[dotstyle=pentagon](1,2)(1.5,3)(2.25,4)(3.375,5)(5.06,6)
      \psdots[dotstyle=pentagon](7.59,7)
      \psdots[dotstyle=pentagon](2,1)(3,2)(4.5,3)
      \psdots[dotstyle=pentagon](6.75,4)(10.125,5)(15.1875,6)
      \psdots[dotstyle=pentagon](3,0)(4.5,1)
      \psdots[dotstyle=pentagon](6.75,2)(10.125,3)(15.1875,4)
      \psgrid(0,0)(0,0)(17,7)
      \psaxes[arrows=->](0,0)(0,0)(17,7)
      \rput(17.7,0){{\small $x$}}
      \rput(0,7.45){{\small $y$}}
    \end{pspicture}
    \caption{Trajectories (dashed) starting from $(1,2)$, $(2,1)$ and $(3,0)$ in Ex.~\ref{ex:running1},
      initial set of states $X$ (dark gray), and invariant
      $Y$ approximating $A^*X$ (light gray) in
      Ex.~\ref{ex:running2}.
    }
    \label{fig:running2}
  }\hfill
  \parbox{0.39\linewidth}{
    \centering
    \psset{unit=1.1em,griddots=5,subgriddiv=1,gridlabels=0pt,Dx=2,labelsep=0.2pt}
    \vspace{1ex}
    \begin{pspicture}(0,-1)(17,5.5)
      \pspolygon[fillstyle=solid,fillcolor=LightGrey,linestyle=none](1,0)(1,2)(1.928,4)(17,4)(17,0)
      \psline(17,0)(1,0)(1,2)(1.928,4)(17,4)
      \psdots(1,0)(1,2)(1.928,4)
      \psline[arrowsize=3pt 4,arrowinset=0,arrows=->](3,0)(5,0)
      \pspolygon[fillstyle=solid,fillcolor=DarkGrey](1,0)(1,2)(1.5,3.75)(1.875,4)(15.19,4)(15.19,3.375)(3,0)
      \psdots(1,0)(1,2)(1.5,3.75)(1.875,4)(15.19,4)(15.19,3.375)(3,0)
      \psframe[fillstyle=solid,fillcolor=darkgray](1,0)(3,2)
      \psline[linewidth=0.5pt,linestyle=dashed](1,2)(1.5,3)(2.25,4)
      \psline[linewidth=0.5pt,linestyle=dashed](2,1)(3,2)(4.5,3)(6.75,4)
      \psline[linewidth=0.5pt,linestyle=dashed](3,0)(4.5,1)(6.75,2)(10.125,3)(15.1875,4)
      \psdots[dotstyle=pentagon](1,2)(1.5,3)(2.25,4)
      \psdots[dotstyle=pentagon](2,1)(3,2)(4.5,3)
      \psdots[dotstyle=pentagon](6.75,4)
      \psdots[dotstyle=pentagon](3,0)(4.5,1)(6.75,2)(10.125,3)(15.1875,4)
      \psgrid(0,0)(0,0)(17,5)
      \psaxes[arrows=->](0,0)(0,0)(17,5)
      \rput(17.7,0){{\small $x$}}
      \rput(0,5.45){{\small $y$}}
      \psline[linestyle=dashed](-0.5,3)(17.5,3)
      \rput(15.5,2.3){$\mathbf{y\leq 3}$}
    \end{pspicture}

    \caption{Initial set of states $X$ (dark gray), invariant
      $Y$
      approximating $(G\srightarrow A)^*(X)$
      in Ex.~\ref{ex:running3} (light gray), and the better
      invariant $Z$ (medium gray)
      discovered in Ex.~\ref{ex:running4} by exploiting the number of iterations.
    }
    \label{fig:running3}
  }
\end{figure*}

\paragraph{Abstract acceleration without guards using
  matrix abstract domains}
\label{sec:overview:noguard}
We  now present an overview for linear loop without guards, with semantic
function $\tau=(\true\rightarrow A)$. For any set $X$, we have
$$
\tau^*(X)=\bigcup_{n\geq 0} \tau^n(X) = \bigcup_{n\geq 0} A^n X
$$ Our approach computes a finitely representable approximation
$\scrM$ of the countably infinite set of matrices $\bigcup_{n\geq 0}
A^n$. Thereafter, abstract acceleration simply applies $\scrM$ to $X$.

The following example illustrates the first step.
\begin{example}[Exponential 1/4]\label{ex:running1}
  We consider the program
  \begin{quote}\tt
    while(true)\{ x=1.5*x; y=y+1 \}
  \end{quote}
  of which Fig.~\ref{fig:running1} depicts some trajectories.
  After homogenization, the loop's semantic function is
  \vspace*{-0.4ex}
  $$
  G=
  \begin{pmatrix}
      0 & 0 & 0
    \end{pmatrix}
  \rightarrow
  A=
    \begin{pmatrix}
      1.5 & 0 & 0 \\
      0 & 1 & 1 \\
      0 & 0 & 1
    \end{pmatrix}
    \vspace*{-0.4ex}
  $$
  Here, it is easy to obtain a closed-form symbolic
  expression of $A^n$:
  \vspace*{-0.4ex}
  $$
  A^n=
  \begin{pmatrix}
    1.5^n & 0 & 0 \\
    0 & 1 & n \\
    0 & 0 & 1
  \end{pmatrix}
  \vspace*{-0.4ex}
  $$
  The idea for approximating $\bigcup_{n\geq 0} A^n$ is to
  consider a set of matrices of the form %
  \vspace*{-0.4ex}
  $$
  \scrM=\left\{
    \begin{pmatrix}
      m_1 & 0 & 0 \\
      0 & 1 & m_2 \\
      0 & 0 & 1
    \end{pmatrix} \left|\;
      \varphi_\scrM(m_1,m_2)
    \right.
  \right\}
  \vspace*{-0.4ex}
  $$
  with $\varphi_\scrM$ a linear assertion in a template domain
  such that $\forall n\sgeq 0: A^n\subseteq \scrM$.
 Using an octagonal template, for instance, the following assertion satisfies the condition above:
  \vspace*{-0.4ex}
  \begin{equation*}
    \varphi_M:\left\{\begin{array}{*{5}{@{\,}c}@{\,}}
	m_1 &\in&  [ 1, +\infty] &=& [\inf\limits_{n\geq 0} 1.5^n, \sup\limits_{n\geq 0} 1.5^n]\\
	m_2 &\in& [0,+\infty] &=& [\inf\limits_{n\geq 0} n, \sup\limits_{n\geq 0} n] \\
	m_1\sadd m_2 &\in& [1,+\infty] &=& [\inf\limits_{n\geq 0} (1.5^n\sadd n), \sup\limits_{n\geq 0} (1.5^n\sadd n)] \\
	m_1\ssub m_2 &\in& [0.25,+\infty] &=& [\inf\limits_{n\geq 0} (1.5^n\ssub n), \sup\limits_{n\geq 0} (1.5^n\ssub n)]
      \end{array}\right.
    \vspace*{-0.4ex}
  \end{equation*}
  These constraints actually define the smallest octagon on entries
  $m_1,m_2$ that makes $\scrM$ an overapproximation of $A^*\seq\{ A^n
  \,|\, n\sgeq 0\}$. It is depicted in Fig.~\ref{fig:running1}. The
  technique to evaluate the non-linear $\inf$ and $\sup$ expressions
  above is described in \S\ref{sec:accelerating:bound}.
\end{example}
This is the first important idea of the paper.
\S\ref{sec:matrix} formalizes the notion of abstract matrices, whereas \S\ref{sec:accelerating} will exploit the
Jordan normal form of $A$ to effectively compute $\alpha(A^*)$ for
any matrix $A$, \ie~to accelerate the loop body.

\paragraph{Applying the abstraction to acceleration}
\label{sec:overview:matmul}
The next step is to apply the matrix abstraction $\scrM=\alpha(A^*)$
to an abstract element $X$. For illustration, assume that both $\scrM$
and $X$ are defined by linear assertions $\varphi_\scrM$ and
$\varphi_{X}$ from the polyhedral domain or some sub-polyhedral
domains. Applying the set of matrices $\scrM$ to $X$ amounts to
computing (an approximation of)
\vspace*{-0.4ex}
\begin{equation}\label{eq:running:matmul}
\left\{
  \begin{pmatrix}
    m_1 & 0 & 0 \\
    0 & 1 & m_2 \\
    0 & 0 & 1
  \end{pmatrix}
  \begin{pmatrix}
    x \\ y \\ 1
  \end{pmatrix}
  \;\Bigg|
  \begin{array}{c@{\,}c@{\,}ll}
    \varphi_{\scrM}(m_1,m_2) & \wedge \\
    \varphi_{X}(x,y)
  \end{array}
\right\}
\vspace*{-0.4ex}
\end{equation}
This is not trivial, as the matrix multiplication generates bilinear
expressions.  \S\ref{sec:matrix} proposes a general approach for
performing the abstract multiplication. The result of the procedure
is illustrated by the example that follows:

\begin{example}[Exponential 2/4]\label{ex:running2}
  Assume that in Ex.~\ref{ex:running1} and
  Eqn.~\eqref{eq:running:matmul}, $\varphi_X=(x\in[1,3]\wedge
  y\in[1,2])$. We compute the abstract matrix multiplication
  \vspace*{-0.4ex}
  \begin{align*}
    \scrM X =&
    \left\{
      \begin{pmatrix}
	m_1\cdot x \\
	1 + m_2\cdot y \\
	1
      \end{pmatrix}
      \;\Bigg|
      \begin{array}{ll}
      1\sleq x\sleq 3\wedge 0\sleq y\sleq 2 & \wedge \\
      m_1\sgeq 1 \wedge m_2\sgeq 0 & \wedge \\
      m_1\ssub m_2\sgeq 0.25
    \end{array}
    \right\} \\
  \end{align*}
  \begin{align*}
    \subseteq&
    \left\{
      \begin{pmatrix}
	x' \\
	y' \\
	1
      \end{pmatrix}
      \;\Bigg|
      \begin{array}{ll}
      x'\sgeq 1 \wedge y'\sgeq 0 & \wedge \\
      x'\ssub 1.5y'\sgeq-4.125 & \wedge \\
      3.5x'\ssub y' \sgeq 1.5
    \end{array}
    \right\} = Y
    \vspace*{-0.4ex}
  \end{align*}
  where $Y$ is the result obtained by the method described in
  \S\ref{sec:matrix} and is depicted in Fig.~\ref{fig:running2}.
\end{example}

\paragraph{Handling Guards}
\label{sec:overview:guard}
We consider loops of the form $\tau=G\rightarrow A$ and illlustrate
how the loop condition (guard) $G$ is handled.  A simple approach
takes the guard into account after the fixpoint of the loop without
guard is computed:
\vspace*{-0.4ex}
$$
(G\rightarrow A)^*(X) \;\subseteq\; X \;\cup\; (G\rightarrow A)\circ
(G\rightarrow A^*)(X)
\vspace*{-0.4ex}
$$
which is then abstracted with $X\sqcup (G\rightarrow
A)\circ(G\rightarrow \alpha(A^*))(X)$.
However, such an approach is often unsatisfactory.
\begin{example}[Exponential 3/4]\label{ex:running3}
  We add the guard $y\leq 3$ to our running
  Ex.~\ref{ex:running1}. Using the approximation above with $X$ as
  in Ex.~\ref{ex:running2}, we obtain the invariant $Y$
  depicted in Fig.~\ref{fig:running3}.  In this result $y$ is
  bounded but $x$ remains unbounded. 
\end{example}

Our idea is based on the observation that the bound on $y$ induced by the
guard implies a bound $N$ on the maximum number of iterations for
any initial state in $X$. Once this bound is known, we can exploit
the knowledge $\tau^*(X)=\bigcup_{n=0}^{N}\tau^n(X)$ and consider
the better approximation
\vspace*{-0.4ex}
$$
\textstyle
(G\rightarrow A)^*(X)\subseteq X\sqcup (G\rightarrow
A)\circ\left(G\rightarrow \alpha\left(\bigcup\limits_{n=0}^{N-1} A^n\right)\right)(X)
\vspace*{-0.4ex}
$$

The set of matrices $\bigcup\limits_{n=0}^{N-1} A^n$ is then
approximated in the same way as $A^*$ in
\S\ref{sec:overview:noguard}. We could perform an iterative
computation for small $N$, however, a polyhedral analysis without
widening operator is intractably expensive for hundreds or thousands
iterations, while our method is both, precise and efficient.

\begin{example}[Exponential 4/4]\label{ex:running4}
  In our running example, it is easy to see that the initial
  condition $y_0\in[0,2]$ together with the guard $y\sleq 3$
  implies that the maximum number of iteration is $N=4$.  Thus, we
  can consider the set of matrices
  $\scrM'=\alpha(\bigcup\limits_{n=0}^{N-1} A^n)$ defined by the
  following assertion satisfies the condition $\varphi_{\scrM'}$ above:
  \begin{equation*}\label{ex:running:octagon}
    \left\{\begin{array}{*{5}{@{\,}c}@{\,}}
	m_1 &\in&  [ 1, 3.375] &=& [\inf\limits_{0\leq n\leq 3} 1.5^n, \sup\limits_{0\leq n\leq 3} 1.5^n]\\
	m_2 &\in& [0,3] &=& [\inf\limits_{0\leq n\leq 3} n, \sup\limits_{0\leq n\leq 3} n] \\
	m_1+m_2 &\in& [1,6.375] &=& [\inf\limits_{0\leq n\leq 3} (1.5^n\sadd n), \sup\limits_{0\leq n\leq 3} (1.5^n\sadd n)] \\
	m_1-m_2 &\in& [0.25,1] &=& [\inf\limits_{0\leq n\leq 3} (1.5^n\ssub n), \sup\limits_{0\leq n\leq 3} (1.5^n\ssub n)]
      \end{array}\right.
  \end{equation*}
  which is depicted in Fig.~\ref{fig:running1}.  Using the formula
  above, we obtain the invariant $Z$ depicted in
  Fig.~\ref{fig:running3}, which is much more precise than the
  invariant $Y$ discovered with the simple technique.
\end{example}
\S\ref{sec:guard} presents the technique for over-approximating the
number of iterations of a loop where the guard $G$ is a general linear
assertion, $A$ is any matrix and $X$ any polyhedron. This is the third main
contribution of the paper.

\paragraph{An Illustrative Comparison}
\label{sec:summary-exp}
The capability of our method to compute over-approximations of
the reachable state space goes beyond state-of-the-art invariant inference techniques.
The following table lists the bounds obtained on the variables of
the thermostat of
Ex.~\ref{Ex:iir-filter-example}, Fig.~\ref{fig:thermostat} for
some competing techniques:

\smallskip\noindent
\begin{tabular}{|@{}c@{}|@{}c@{}|@{}c@{}|@{}c@{}|}
\hline
\textsc{InterProc} \cite{interproc} & \textsc{Astr\'ee} \cite{BCC+03} &
\textsc{Sting} \cite{CSS03,SSM04} & this paper \\
\hline
\multicolumn{4}{|l|}{heating:} \\
\hline
$16\sleq t$& $16\sleq t\sleq 22.5$&
$16\sleq t\sleq 22.5$ & $16\sleq t\sleq 22.5$\\
$0\sleq \mathit{time}$ & $0\sleq
\mathit{time}$ & $0\sleq \mathit{time}\sleq 13$ & $0\sleq \mathit{time}\sleq 9.76$ \\
\hline
\multicolumn{4}{|l|}{cooling:} \\
\hline
$t\sleq 22.5$ & $17.75\sleq t\sleq 22.5$ & $17.75\sleq t\sleq 22.5$ & $17.75\sleq t\sleq 22.5$\\
$0\sleq \mathit{time}$& $0\sleq
\mathit{time}$& $0\sleq \mathit{time}\sleq 19$ & $0\sleq \mathit{time}\sleq 12.79$\\
\hline
\end{tabular}

\smallskip
There are many other invariant generation techniques and tools for
linear systems (see \S\ref{sec:rw}). Many approaches sacrifice
precision for speed, and therefore are inaccurate on the type of
linear loops considered here. Other, more specialized approaches
require conditions such as Lyapunov-stability, diagonalizability of the
matrix, polynomial behavior (nilpotency or monoidal property), or
handle only integer loops.

\paragraph{Outline of the rest of the paper}
\label{sec:outline}

The rest of the paper develops the ideas illustrated in this section.
\S\ref{sec:matrix} formalizes the notion of matrix abstract domains
and presents a technique for the abstract matrix multiplication
operation.  \S\ref{sec:accelerating} shows how to approximate the set
of matrices $A^*=\bigcup_{n\geq 0} A^n$ for any square matrix $A$ in
order to accelerate loops without guards.  \S\ref{sec:guard} presents
a technique for taking the guard of loops into account by
approximating $N$, the maximum number of iterations possible from a
given set of initial states. \S\ref{sec:exp} presents the experimental
evaluation on various kinds of linear loops, possibly embedded into
outer loops. \S\ref{sec:rw} discusses related work and
\S\ref{sec:concl} concludes.

\section{Matrix abstract domains}
\label{sec:matrix}

In this section we present abstract domains for matrices.  We will use
abstract matrices to represent the accelerated abstract transformer of
a linear loop.  Hence, the main operation on abstract matrices we use
in this paper is abstract matrix multiplication (\S\ref{sec:matmul}).

\subsection{Extending abstract domains from vectors to matrices}
\label{sec:matrix:def}

We abstract sets of square matrices in $\reals^{p\times p}$
by viewing them as vectors in $\reals^{p^2}$ and
by reusing known abstract domains over vectors.
However, since the concrete matrices we will be dealing with
belong to subspaces of $\reals^{p\times p}$, we first introduce
\emph{matrix shapes} that allow us to reduce the number of entries
in abstract matrices.
\begin{definition}[Matrix shape]
  A \emph{matrix shape} $\Psi:\reals^m \rightarrow \reals^{p\times p}$
  is a bijective, linear map from $m$-dimensional vectors to $p \times
  p $ square matrices. Intuitively, matrix shapes represent matrices
  whose entries are linear (or affine) combinations of $m > 0$
  entries.

\end{definition}
\begin{example}\label{ex:running1-matrixshape}
  In Ex.~\ref{ex:running1}, we implicitly considered the
  matrix shape
  $
    \begin{array}[t]{@{\,}c@{\,}c@{\,}c@{\,}c@{\,}c@{\,}}
      \Psi: & \reals^2 &\rightarrow& \mathit{Mat}_\Psi & \subseteq \reals^{3\times 3} \\
      & \begin{pmatrix} m_1 \\ m_2 \end{pmatrix}
      &\mapsto&
      \begin{pmatrix}
        m_1 & 0 & 0 \\
        0 & 1 & m_2 \\
        0 & 0 & 1
      \end{pmatrix}
    \end{array}
  $
\end{example}

The set $\mathit{Mat}_\Psi=\{ \Psi(\vm) \;|\; \vm\in \reals^m\}$
  represents all possible matrices that can be formed by any vector
  $\vm$. It represents a subspace of the vector space of all matrices.

A matrix shape $\Psi$ induces an isomorphism between $\reals^m$ and
$\mathit{Mat}_\Psi$: $\Psi(a_1\vm_1+ a_2\vm_2) = a_1 \Psi(\vm_1) + a_2
\Psi(\vm_2)$.

Abstract domain for matrices are constructed by (a) choosing
an abstract domain for vectors $\vm$ and (b) specifying a matrix
shape $\Psi$. Given an abstract domain $A$ for vectors $\vm$
(eg, the polyhedral domain) and a shape $\Psi$, the corresponding
matrix abstract domain defines a domain over subsets of $\mathit{Mat}_{\Psi}$.

\begin{example}\label{ex:running1-matrixdomain}
Recall the matrix shape $\Psi: \reals^2 \rightarrow \mathit{Mat}_\Psi
\subseteq \reals^{3\times 3}$ from Ex.~\ref{ex:running1-matrixshape}.
Consider the octagon $P=(m_1\sgeq 1 \wedge m_2\sgeq 0 \wedge m_1\sadd
m_2\sgeq 1 \wedge m_1\ssub m_2\sgeq 0.25) \in \mathit{Oct}(\reals^2)$.
Together they represent an abstract matrix $(P,\Psi)$ which
represents the set of matrices:
  $$ \left\{
    \begin{pmatrix}
      m_1 & 0 & 0 \\
      0 & 1 & m_2 \\
      0 & 0 & 1
    \end{pmatrix} \left|
    \begin{array}{c}
      m_1\sgeq 1,  m_2\sgeq 0, \\ m_1\sadd m_2\sgeq 1, \\ m_1\ssub m_2\sgeq 0.25
    \end{array} \right.
  \right\}.
  $$
\end{example}

\begin{definition}[Abstract domain for matrices induced by $\Psi$]\label{def:APsi} ~\\ %
  Let $A\subseteq \wp(\reals^m)$ be an abstract domain
  for $m$-dimensional vectors ordered by set inclusion and with the
  abstraction function
  $\alpha_A:\wp(\reals^m)\rightarrow A$.
  Then, $\Psi(A)$ ordered by set inclusion
  is an abstract domain for $\wp(\mathit{Mat}_\Psi)$ with the
  abstraction function
  \vspace*{-0.4ex}
  \begin{align*}
    \alpha_{\Psi(A)}(\scrM) &=
    \Psi\scirc\alpha_A\scirc\Psi^{-1}(\scrM) \,.
  \vspace*{-0.4ex}
  \end{align*}

\end{definition}

Note that since $\Psi$ is an isomorphism: the lattices $A$ and
$\Psi(A)$ can be shown to be isomorphic.  For generality, the
\emph{base domain} $A$ can be an arbitrary abstract domain for the
data type of the matrix entries. In our examples, we specifically
discuss common numerical domains such as convex polyhedra, intervals,
octagons and templates.

\subsection{Abstract matrix multiplication}
\label{sec:matmul}

We investigate now the problem of convex polyhedra matrix
multiplication, motivated by the need for applying an
acceleration $\alpha(A^*)$ to an abstract property $X$ as
shown in \S\ref{sec:overview:matmul}.

\paragraph{The problem.}
We consider two convex polyhedra matrices
$\scrM_s=\Psi_s(P_s)$. We aim at computing an approximation of
\begin{equation}
  \label{eq:matmul}
  \scrM=\scrM_1\scrM_2=\{ M_1M_2 \;|\; M_1\in\scrM_1 \wedge  M_2\in\scrM_2 \}
\end{equation}
under the form of a convex polyhedron on the coefficients of the
resulting matrix. Observe that $\scrM$ may be non-convex as shown
by the following example.
\begin{example}\label{ex:matrix-matrix-multiplication-1}
Consider the two abstract matrices
\vspace*{-0.4ex}
$$
\begin{array}{@{}l@{}}
  \scrM_1\seq\left \{
    \left(
    \begin{array}{@{}c@{\,}c@{}}
      1\ssub m & 0\\
      0 & m
    \end{array}\right)
    \Big|\,
    m\in[0,1]
  \right\}
  \;
  \scrM_2\seq
  \left\{
    \begin{pmatrix}
      1\ssub n \\
      n
    \end{pmatrix}
    \Big|\,
    n\in[0,1]
  \right\}
\end{array}
\vspace*{-0.4ex}
$$
We have
\vspace*{-0.4ex}
$$
\scrM_1\scrM_2=\left\{
\begin{pmatrix}
(1-m)(1-n) \\
mn
\end{pmatrix}
\Big|\;
m\in[0,1] \wedge n\in[0,1]
\right\}
\vspace*{-0.4ex}
$$
\parbox{0.64\linewidth}{
\noindent This corresponds to the well-known non-convex set of points
$x\in[0,1]\wedge y\in[0,1]\wedge (x\ssub y)^2\sadd 1\sgeq 2(x\sadd
y)$, depicted to the right.
}\hfill\parbox{0.33\linewidth}{
    \centering
    \psset{unit=6em,griddots=5,subgriddiv=5,gridlabels=7pt}

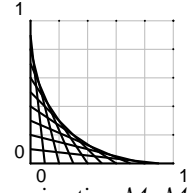
    \begin{pspicture}(0,0)(1,1)
      \psgrid(0,0)(0,0)(1,1)
      \psline(0,1)(0,0)
      \psline(0,0.9)(0.1,0)
      \psline(0,0.8)(0.2,0)
      \psline(0,0.7)(0.3,0)
      \psline(0,0.6)(0.4,0)
      \psline(0,0.5)(0.5,0)
      \psline(0,0.4)(0.6,0)
      \psline(0,0.3)(0.7,0)
      \psline(0,0.2)(0.8,0)
      \psline(0,0.1)(0.9,0)
      \psline(0,0)(1,0)
    \end{pspicture}
}\null
\vspace*{0.7ex}
\end{example}
We may follow at least two approaches for approximating $\scrM_1\scrM_2$:
\begin{itemize}
\item Either we consider the constraint representations of
  $\scrM_1$ and $\scrM_2$, and we resort to optimization techniques
  to obtain a template polyhedra approximation of the product;
\item Or we consider their generator representations to obtain a
  convex polyhedron approximating the product.
\end{itemize}
We opted in this paper for the second, \ie~the \emph{generator} approach,
which leads to more accurate results:
\begin{compactitem}[--]
\item it delivers general convex polyhedra, more expressive than
template polyhedra obtained by optimization;
\item it computes the best correct approximation in the convex
  polyhedra domain for bounded matrices (Thm.~\ref{thm:matmul}
  below), whereas in the constraint approach the exact
  optimization problem involves bilinear expressions (see
  Eqn.~(\ref{eq:running:matmul}) or
  Ex.~\ref{ex:matrix-matrix-multiplication-1}) and must be
  relaxed in practice.
\end{compactitem}




\paragraph{Multiplying abstract matrices using generators.}
Given two finite sets of matrices $X = \{ X_1,\ldots,X_m\}$ and $Y=\{ Y_1,\ldots,Y_k\}$
we write $X \otimes Y $  to denote the set
$$
X \otimes Y = \{ X_i Y_j\ |\ X_i \in X,\ Y_j \in Y \}
$$
If $\scrM_s$ is expressed as a system of
matrix vertices $V_s=(V_{s,i_s})$ and matrix rays
$R_s=(R_{s,j_s})$, $s=1,2$, then
\vspace*{-0.4ex}
\begin{multline*}\textstyle
\scrM_s =
\Bigl\{ \sum_{i_s} \lambda_{i_s} V_{s,i_s} + \sum_{j_s}\mu_{j_s}
R_{s,j_s} \;\Big|\;
\begin{array}{@{}l@{}}
  \lambda_{i_s},\mu_{j_s}\geq 0 \\
  \sum_{i_s}\lambda_{i_s}=1
\end{array}
\Bigr\}
\end{multline*}
and Eqn.~(\ref{eq:matmul}) can be
rewritten
\vspace*{-0.4ex}
\begin{multline}\label{eq:matmulgen}
  \scrM=\scrM_1\scrM_2 = \\
  \left\{
    \begin{array}{@{\,}l@{\,}}
      \sum_{i_1,i_2} \lambda_{1,i_1}\lambda_{2,i_2}
      V_{1,i_1} V_{2,i_2} \\
      +\sum_{i_1,j_2} \lambda_{1,i_1}\mu_{2,j_2}
      V_{1,i_1} R_{2,j_2} \\
      +\sum_{i_2,j_1} \mu_{1,j_1}\lambda_{2,i_2}
      R_{1,j_1} V_{2,i_2} \\
	+\sum_{j_1,j_2} \mu_{1,j_1}\mu_{2,j_2}
	R_{1,j_1} R_{2,j_2}
      \end{array}
      \left|
     \begin{array}{@{\,}l@{\,}}
       \lambda_{1,i_1}, \mu_{1,j_1}\geq 0 \\
       \lambda_{2,i_2}, \mu_{2,j_2}\geq 0 \\
       \sum_{i_1} \lambda_{1,i_1}=1 \\
       \sum_{i_2} \lambda_{2,i_2}=1 \\
      \end{array}
      \right.\right\}
  \end{multline}
We obtain the following result:
\begin{theorem}\label{thm:matmul}
Let $\scrM_1$ and $\scrM_2$ be two abstract matrices expressed as a
system of vertices and rays: $V_1, R_1$ for $\scrM_1$ and
$V_2, R_2$ for $\scrM_2$.  The matrix polyhedron $\widetilde{\scrM}$
defined by the set of vertices
\centerline{$V = V_1 \otimes V_2\vspace*{-0.4ex}$}
and the set of rays\\
\centerline{$R = (V_1 \otimes R_2) \cup (R_1 \otimes V_2) \cup (R_1 \otimes R_2)$}
is an overapproximation of $\scrM=\scrM_1\scrM_2$.

Moreover, if $\scrM_1$ and $\scrM_2$ are bounded, \ie~ if
$R_1=R_2=\emptyset$, then $\widetilde{\scrM}$ is the smallest
polyhedron matrix containing $\scrM$.
\end{theorem}
\begin{proof}
  For the first part of the theorem, we observe that
  in Eqn.~(\ref{eq:matmulgen}), $\sum_{i_1,i_2} \lambda_{1,i_1}\lambda_{2,i_2} = 1$ and the
  other similar sums are positive and unbounded. Hence \\[0.5ex]
  $
  \scrM\subseteq \widetilde{\scrM}=
  \left\{
    \begin{array}{@{\,}l@{\,}}
      \sum_{i_1,i_2} \lambda'_{i_1,i_2}
      V_{1,i_1} V_{2,i_2} \\[0.5ex]
      +\sum_{i_1,j_2} \mu'_{i_1,j_2}
      V_{1,i_1} R_{2,j_2} \\[0.5ex]
      +\sum_{i_2,j_1} \mu''_{j_1,i_2}
      R_{1,j_1} V_{2,i_2} \\[0.5ex]
      +\sum_{j_1,j_2} \mu'''_{j_1,j_2}
      R_{1,j_1} R_{2,j_2}
    \end{array}
    \left|
      \begin{array}{@{\,}l@{\,}}
        \lambda'_{i_1,i_2}\geq 0 \\[0.5ex]
        \mu'_{i_1,j_2}\geq 0 \\[0.5ex]
        \mu''_{j_1,i_2}, \mu'''_{j_1,j_2}\geq 0 \\[0.5ex]
        \sum_{i_1,i_2} \lambda'_{i_1,i_2}=1
      \end{array}
    \right.
  \right\}
  $ \\[0.5ex]
  which proves the first statement.
  Now assume that $R_1=R_2=\emptyset$, which means that both
  $\scrM_1$ and $\scrM_2$ are bounded and that $R=\emptyset$.
  We will show that all the generator vertices of $\widetilde{\scrM}$ belong
  to $\scrM$, hence any of their convex combination (\ie~
  any element of $\widetilde{\scrM}$) belongs to the convex closure
  of $\scrM$:
  Consider the generator vertex $V_{1,i_1}V_{2,i_2}\in V$ of
  $\widetilde{\scrM}$. By taking in
  Eqn.~(\ref{eq:matmulgen}) $\lambda_{s,i_s}=1$ and $i'_s\neq
  i_s\implies \lambda_{s,i'_s}=0$ for $s=1,2$, we obtain that $V\in
  \scrM$.
\end{proof}
\begin{example}
  In Ex.~\ref{ex:running2}, we multiplied the unbounded set of matrices
  $\scrM$ depicted in Fig.~\ref{fig:running1} and defined by 2
  vertices and 2 rays by the bounded set of vectors $X$ depicted in
  Fig.~\ref{fig:running2} and generated by 4 vertices, which
  resulted in the convex polyhedra $Y$ depicted in
  Fig.~\ref{fig:running2} which is generated by 3 vertices and 2
  rays (we omit redundant generators).
\end{example}

Regarding complexity, this operation is quadratic
w.r.t. the number of generators, which is itself
exponential in the worst-case w.r.t. the number of constraints. In
practice, we did not face complexity problems in our experiments,
apart from the high-dimensional \textit{convoyCar3} example
described in \S\ref{sec:exp}.

Observe that by using generators and applying Thm.~\ref{thm:matmul},
we lose information about matrix shapes. In our case, we will perform
only abstract matrix-vector multiplication, hence the number of
entries of the product matrix (actually a vector) will be the
dimension of the space $\reals^p$.  The multiplication of an abstract
matrix $\mathcal{M}$ and a concrete matrix $R$ ($\scrM R$ or $R\scrM$)
can be computed exactly by considering the generators of
$\mathcal{M}$.

\section{Abstract acceleration of loops without guards}
\label{sec:accelerating}
In this section, we consider loops of the form\\[0.2ex]
\centerline{$\mathsf{while}(\mbox{true})\{ \vx  := A \vx \}$}
Given an initial set of states $X$ at the loop head, the least
inductive invariant at loop head is\\[0.2ex]
\centerline{$A^*X=\{ A^n X \;|\; n \geq 0, n \in \mathbb{Z} \} \,.$}
Our goal is to compute a template polyhedra matrix $\scrM$ such
that\\[0.2ex]
\centerline{$\alpha_T(A^*) \sqsubseteq \scrM$}
given a template $T$ on the coefficients of the matrices $M\in A^*$.

The key observation underlying our approach uses a well-known
result from matrix algebra. Any square matrix $A$ can be written in a
special form known as the Jordan normal form using a change of
basis transformation $R$:\\[0.2ex]
\centerline{$A = R^{-1} J R \text{ and } J = R A R^{-1}$}
such that for any $n$ \\[0.2ex]
\centerline{$A^n = R^{-1} J^n R \text{ and } J^n = R A^n R^{-1} \,.$}
As a result, instead of computing an abstraction of the set\\[0.2ex]
\centerline{$A^* = \{ I, A, A^2, A^3, \ldots \}$}
we will abstract the set\\[0.2ex]
\centerline{$J^* = \{ I, J, J^2, J^3, \ldots \}\ .$}
\begin{enumerate}
\item The block diagonal structure of $J$ allows us to
  \emph{symbolically} compute the coefficients of $J^n$ as a
  function of $n$.  \S\ref{sec:accelerating:jordan} presents
  details on the Jordan form and the symbolic representation of~$J^n$.
\item The form of $J$ immediately dictates the matrix shape
  $\Psi(\vm)$ and the matrix subspace $\mathit{Mat}_\Psi$
  containing $J^*$.
\item We then consider a fixed set $T$ of \emph{linear template
  expressions} over $\vm$. We use asymptotic analysis to compute
  bounds on each expression in the template.
  \S\ref{sec:accelerating:bound} explains how this is computed.
\item Once we have computed an abstraction $\scrM \sqsupseteq
  \alpha(J^*)$, we will return into the original basis by computing
  $R^{-1}\scrM R$ to obtain an abstraction for $A^*$, which is the
  desired loop acceleration.
\end{enumerate}
In this section, we assume arbitrary precision numerical
computations.  The use of finite precision computations is addressed
in \S\ref{sec:exp}.



\subsection{The real Jordan normal form of a matrix}
\label{sec:accelerating:jordan}

A classical linear algebra result is that any matrix
$A\in\reals^{p\times p}$
can be put in a \emph{real Jordan normal form} by considering an
appropriate basis \cite{LT84}:
\begin{gather*}
A=R^{-1}
\overbrace{\begin{pmatrix}
  J_1 \\
  & \ddots \\
  && J_r
\end{pmatrix}
}^{\mathclap{\displaystyle J= \diag[J_1\ldots J_r]}}
R
\; , \;
J_s \defeq
\bordermatrix{
  & \scriptstyle 0 & \scriptstyle 1 && \scriptstyle p_s-1 \cr
  & \Lambda_s & I \cr
  && \ddots & \ddots \cr
  &&& \Lambda_s & I \cr
  &&&& \Lambda_s
} \\
\begin{array}{@{}r@{\,}l@{}}
\text{with}\quad
\Lambda_s &=\lambda_s \text{ and } I=1 \\[1ex]
& \text{if $\lambda_s$ is a real eigenvalue of $M$,} \\[0.5ex]
\text{or}\quad
\Lambda_s &=
  \rpmatrix{\lambda_s\cos \theta_s}{\lambda_s\sin \theta_s} \text{
    and }
  I =
  \begin{pmatrix} 1 & 0 \\ 0 & 1 \end{pmatrix} \\[2ex]
  & \parbox[t]{20em}{if $\lambda_se^{i\theta_s}$ and
    $\lambda_s e^{-i\theta_s}$ are complex
    conjugate eigenvalues of $M$, with $\lambda_s>0$ and $\theta_s\in[0,\pi[$.}
\end{array}
\end{gather*}
The Jordan form is useful because we can write a closed-form expression
for its $n^\text{th}$ power $J^n=\diag[J_1^n\ldots J_r^n]$. Each block $J_s^n$
is given by
\begin{gather}\label{eq:Jsn}
J_s^n = \begin{pmatrix}
  \Lambda_s^n & \tbinom{n}{1}\Lambda_s^{n-1} & \hdotsfor{2} & \tbinom{n}{p_s-1}\Lambda_s^{n-p_s+1}\\
	      & \Lambda_s^n                 & \tbinom{n}{1}\Lambda_s^{n-1} && \vdots    \\
	      &                             & \ddots       && \tbinom{n}{1}\Lambda_s^{n-1} \\
	      &                             &              && \Lambda_s^n
\end{pmatrix} \\
\text{and }\Lambda_s^n =
  \lambda_s^n\begin{pmatrix}1\end{pmatrix}
  \quad \text{or} \quad
  \lambda_s^n\rpmatrix{\cos n\theta_s}{\sin n\theta_s}
  \notag
\end{gather}
with the convention that
$\tbinom{n}{k}\lambda_s^{n-k}=0$ for $k>n$. \smallskip

Hence, coefficients of $J_s^n$ have the general form
\begin{equation}\label{eq:varphi}
\varphi[\lambda,\theta,r,k](n)=\tbinom{n}{k}\lambda^{n-k}
\cos((n-k)\theta-r\tfrac{\pi}{2})
\end{equation}
with $\lambda\sgeq 0$, $\theta\in[0,\pi]$, $r\in\{0,1\}$ and
$k\sgeq 0$, in which $r\seq 1$ enables converting the cosine into a sine.
The precise expressions for $\lambda, \theta, r, k$ as functions
of the position $i,j$ in the matrix $J$ are omitted here to
preserve the clarity of presentation. \smallskip

Next, we observe that the closed form $J_s^n$ specifies the
required shape $\Psi_s(\vm_s)$ for abstracting $J_s^n$ for all $n
\geq 0$. For instance, if $\lambda_s$ is a real eigenvalue, we
have
\vspace*{-0.4ex}
$$
\Psi_s:
\begin{spmatrix}
  m_0\\
  \vdots\\
  m_{p_s-1}
\end{spmatrix}
\mapsto
\begin{spmatrix}
  m_0 & m_1 & \hdotsfor{1} & m_{p_s}-1 \\
      & m_0 & m_1 & \vdots    \\
      &     & \ddots & m_1 \\
      &     &  & m_0
\end{spmatrix}
\vspace*{-0.4ex}
$$
Likewise, $\Psi(\vm)$ for the entire matrix $J^n$ is obtained by the
union of the parameters for each individual $J_s^n$.
\begin{proposition}
Given the structure of the real Jordan normal form $J$, we may fix a matrix shape
$\Psi(\vm)$ such that $J^* = \{ I, J, J^2, \ldots\}\subseteq
\mathit{Mat}_\Psi$ and $\vm\in\reals^m$ with $m\leq p$, where $p$
is the dimension of the square matrix $J$.
\end{proposition}
Hence, we will work in a matrix subspace, the dimension of which is
less than or equal to the number of variables in the loop.
This  reduction of dimensions using matrix shapes is absolutely
essential for our technique to be useful in practice.


\subsection{Abstracting $J^*$ within template polyhedron matrices}
\label{sec:accelerating:bound}

\paragraph{The principle.}
Let us fix a template expression matrix $T\in\reals^{q\times m}$ composed of
linear expressions $\{ T_1,\ldots, T_q\}$ on parameters
$\vm$.  Knowing the symbolic form of each $J^n$, we obtain a
symbolic form
$
\vec{m}(n)=\Psi^{-1}(J^n)
$
for parameters
$\vec{m}$, hence a symbolic form for linear expressions
$
e_j(n)=T_j\cdot \vec{m}(n) \,.
$
By deriving an upper bound $u_j$ for
each $e_j(n), n\geq 0$, we obtain a sound approximation of the set $\{
\vec{m}(n)\;|\; n\sgeq 0 \}$, and hence of $J^*$.
\begin{theorem}[Abstracting $J^*$ in template polyhedron
  matrices]
  The template polyhedron matrix
  \vspace*{-0.4ex}
  $$
  \alpha_T(J^*)\defeq \Psi(\denotation{T\vec{m}\leq\vec{u}})
  \text{ with } \vec{u}=\sup\nolimits_{n\geq 0} \,T\vec{m}(n)
  \vspace*{-0.4ex}
  $$
  is the best correct overapproximation of $J^*$ in the template
  polyhedra matrix domain defined by $T$.
  Moreover, any $\vec{u}'\geq\vec{u}$ defines a correct
  approximation of $J^*$.
\end{theorem}
\begin{proof}
  $J^* = \bigcup_{n\geq 0} J^n = \bigcup_{n\geq 0}
  \Psi(\vec{m}(n)) = \Psi(\{\vec{m}(n)\,|\,n\sgeq 0\})$.
  Considering the matrix $T$ and referring to
  \S\ref{sec:polyabsdom},
  \vspace*{-0.4ex}
  \begin{multline*}
    \alpha_T(\{\vec{m}(n)\,|\,n\geq 0\})
    = \{ \vec{m} \;|\; T\vec{m}\leq \sup\limits_{\vec{m}\in
      \{\vec{m}(n)\,|\,n\geq 0\}} T\vec{m} \} \\
    = \{ \vec{m} \;|\; T\vec{m}\leq \sup\limits_{n\geq 0}
    T\vec{m}(n) \}
    =  \{ \vec{m} \;|\; T\vec{m}\leq \vec{u} \}
  \end{multline*}
\end{proof}
The approximation of the set of matrices $J^*$
reduces thus to the computation of an upper bound for the
expressions $e_j(n)$.

\paragraph{Computing upper bounds.}
To simplify the analysis, we restrict template expressions
$T_j$ to involve at most $2$ parameters from $\vm$. As each
parameter/coefficient $m_k$ in matrix $J^n$ is of the form of
Eqn.~(\ref{eq:varphi}), we have to compute an upper bound for
expressions of the form
  \vspace*{-0.4ex}
\begin{multline}\label{eq:phi}
  \mu_1\tbinom{n}{k_1}\lambda_1^{n-k_1}
  \cos((n-k_1)\theta_1-r_1\tfrac{\pi}{2}) \\
  + \mu_2\tbinom{n}{k_2}\lambda_2^{n-k_2}
  \cos((n-k_2)\theta_2-r_2\tfrac{\pi}{2})
  \vspace*{-0.4ex}
\end{multline}
with $\mu_1,\mu_2\in\reals \wedge \mu_1\neq 0$. Computing bounds on
this expression is at the heart of our technique. However, the actual
derivations are tedious and do not contribute to the main insights of
our approach. Hence we omit the detailed derivations, and refer to
\pponly{the extended version \cite{JSS13}}\rronly{appendix \ref{sec:bound}}
for details. The main properties of the technique we implemented are
that it computes
\begin{compactitem}[--]
\item exact bounds if the two involved eigenvalues are
  real ($\theta_1,\theta_2\in\{0,\pi\}$);
\item exact bounds \emph{in reals} if $\theta_1\seq\theta_2 \wedge
  k_1\seq k_2\seq0 \vee \mu_2\seq 0$, and ``reasonable'' bounds if
  $k_1\seq k_2\seq0$ is replaced by $k_1\sgt 0\vee k_2\sgt 0$;
\item no interesting bounds otherwise (because they are just the
  linear combination of the bounds found for each term).
\end{compactitem}
Concerning the choice of template expressions in our
implementation, we fix a parameter $\ell$ and we consider all the
expressions of the form (cf. ``logahedra'' \cite{HK09})
\vspace*{-0.4ex}
\begin{equation}\label{eq:logahedratemplate}
\pm \alpha m_i \pm (1\ssub \alpha) m_j \text{ with }
\alpha\seq\frac{k}{2^\ell}, 0\sleq k\sleq 2^\ell,
\end{equation}
The choice of $\ell=1$ corresponds to octagonal expressions.

\begin{figure}[t]
      \vspace{1ex}
  \psset{unit=5.25em,griddots=1,subgriddiv=4,subgriddots=4,subgridwidth=0.75pt,subgridcolor=black,gridlabels=7pt}
  \hfill\begin{pspicture}(-0.5,-0.25)(1.25,1.0)
    \pspolygon[fillstyle=solid,fillcolor=gray](1,0)(1,0.11)(0.55,0.56)(0,0.56)(-0.29,0.28)(-0.29,0)(-0.14,-0.15)(0.85,-0.15)
    \psgrid(0,0)(-0.5,-0.25)(1.25,1.0)
    \psaxes[labels=none]{->}(0,0)(-0.5,-0.25)(1.25,1.0)
    \rput(1.35,0.1){\small $m_1$}
    \rput(0.15,1.08){\small $m_2$}
  \end{pspicture}
  \hfill
  \psset{unit=1.65em,griddots=5,subgriddiv=1,gridlabels=7pt}
  \begin{pspicture}(-2,-1)(4,3)
    \pspolygon[fillstyle=solid,fillcolor=gray](3,0)(3,2)(1.278,2.878)(1.052,2.988)(-0.4173,2.084)(-1.133,0.4711)(-0.9509,-0.6316)(-0.3767,-0.9081)
    \psframe[fillstyle=solid,fillcolor=darkgray](1,0)(3,2)
    \psgrid(0,0)(-2,-1)(4,3)
    \psaxes[labels=none]{->}(0,0)(-2,-1)(4,3)
    \rput(4.4,0.2){\small $x$}
    \rput(0.3,3.25){\small $y$}
    \psdots[dotstyle=pentagon]
    (3,2)(1.28,2.88)(-0.27,2.04)(-1.00,0.63)(-0.94,-0.44)(-0.48,-0.83)(0,-0.66)(0.27,-0.26)(0.29,0.08)(0.17,0.23)
    \psline[linewidth=0.5pt,linestyle=dashed]
    (3,2)(1.28,2.88)(-0.27,2.04)(-1.00,0.63)(-0.94,-0.44)(-0.48,-0.83)(0,-0.66)(0.27,-0.26)(0.29,0.08)(0.17,0.23)
    \psdots[dotstyle=pentagon]
    (1,2)(-0.11,1.49)(-0.67,0.52)(-0.67,-0.26)(-0.36,-0.57)(-0.02,-0.48)(0.17,-0.21)
    \psline[linewidth=0.5pt,linestyle=dashed]
    (1,2)(-0.11,1.49)(-0.67,0.52)(-0.67,-0.26)(-0.36,-0.57)(-0.02,-0.48)(0.17,-0.21)
    \psdots[dotstyle=pentagon]
    (1,0)(0.69,0.69)(0.20,0.76)(-0.16,0.44)(-0.29,0.06)(-0.22,-0.17)(-0.09,-0.23)(0.03,-0.15)(0.08,-0.04)
    \psline[linewidth=0.5pt,linestyle=dashed]
    (1,0)(0.69,0.69)(0.20,0.76)(-0.16,0.44)(-0.29,0.06)(-0.22,-0.17)(-0.09,-0.23)(0.03,-0.15)(0.08,-0.04)
  \end{pspicture}
  \hfill\null
  \caption{On the left-hand side, the octagon on the two non-constant
    coefficients of the matrices $A^n, n\geq 0$ of
    Ex.~\ref{ex:spiral}, that defines the approximation
    $\scrM\supseteq\{ A^n \;|\; n\geq 0\}$. On the right-hand side,
    the image $\alpha(\scrM X)$ in light gray of the box $X$ in
    dark gray by the set $\scrM$ using the method of \S\ref{sec:matmul}.}
  \label{fig:spiral}
\end{figure}
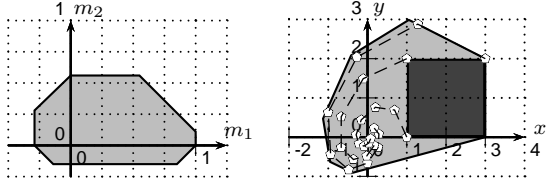

\paragraph{Examples.}
In Ex.~\ref{ex:running2} and Fig.~\ref{fig:running1} we showed the
approximation of a set $A^*$ using octagonal template
expression, with the matrix $A$ being a Jordan normal form with
real eigenvalues 1.5 and 1. 
For instance, consider the expression $m_2\ssub m_1\seq n\ssub 1.5^n$
in Example~\ref{ex:running1}, which falls in the first case above. We
look at the derivative of the function $f(x)\seq x\ssub 1.5^x$, we
infer that $ \forall x\sgeq 3.6:f'(x)\slt 0$ by linearizing
appropriately $f'(t)$, hence we can compute the least upper bound as
$\max\{ f(n)\;|\;0\sleq n\sleq 4\}$.

Next, we give another example with complex eigenvalues.
\begin{example}\label{ex:spiral}
  Take $a=0.8\cos\theta$ and $b=0.8\sin\theta$ with $\theta=\pi/6$.
  and consider the loop
  {\tt
    while(true)\{x'=a*x-b*y; y=a*x+b*y; x=x'\}}.
  The trajectories (see Fig.~\ref{fig:spiral} (right)) of this loop follow
  an inward spiral.
  The loop body transformation is
  $A=
  \begin{spmatrix}
    0.8\cos\theta & -0.8\sin\theta & 0 \\
    0.8\sin\theta & 0.8\cos\theta & 0 \\
    0 & 0 & 1
  \end{spmatrix}
  $ \\[0.5ex]
  with $A$ already in real Jordan normal form.
  The matrix subspace containing $A^*$ is of the form
  $M=\begin{spmatrix}
      m_1 & -m_2 & 0 \\
      m_2 & m_1 & 0 \\
      0 & 0 & 1
    \end{spmatrix}$.
    We have $m_1(n)=0.8^n\cos n\theta$ and $m_2(n)=0.8^n\sin n\theta$;
    applying our bounding technique on
    octagonal template constraints on $\vec{m}$, we obtain an approximation
    $\scrM$ of $A^*$ defined by the constraints \\
  \hspace*{1em}$\left\{\begin{array}{*{3}{@{\,}c}@{\,}@{\quad}*{3}{@{\,}c}@{\,}}
      m_1 &\in& [ -0.29, 1.00] & m_1\sadd m_2 &\in& [-0.29, 1.12]\\
      m_2 &\in& [ -0.15, 0.56 ] & m_1\ssub m_2 &\in& [-0.57,1.00]
    \end{array}\right.
  $ \\
Consider, for example, the expression $m_1\sadd m_2\seq 0.8^n(\cos n\theta\sadd
\sin n\theta)$ in Example~\ref{ex:spiral} below, which falls into the
second case above ($\theta_1\seq\theta_2 \wedge k_1\seq k_2\seq0 \vee \mu_2\seq 0$).
We first rewrite it as
$f(x)=0.8^x\sqrt{2}\sin(x\theta\sadd\frac{\pi}{4})$. The term $0.8^x$
being decreasing, the least upper bound of $f$ \emph{in reals} is in the
range $x\in[0,{\pi}/{4\theta}]$. Hence we can consider the upper bound 
$\max\{ f(n)\;|\; 0\sleq n\sleq \ceil{\pi/4\theta}\}$.

  The possible values $(m_1,m_2)$ are plotted in Fig.~\ref{fig:spiral}
  (right). Assuming an initial set $X:\ x\in[1,3]\wedge y\in[0,2]$, we
  compute $\alpha(\scrM X)$ to be the polyhedron depicted in
  Fig.~\ref{fig:spiral} (right).
\end{example}

In this section, we have described the computation of a
correct approximation $\mathcal{M}$ of $J^*$ in the template
polyhedron domain, from which we can
deduce a correct approximation $R^{-1}\mathcal{M}R$ of
$A^*=R^{-1}J^*R$. Applying Thm.~\ref{thm:matmul}, we are thus able to
approximate the set $A^*X$ of reachable states at the head of a linear loop
$\mathsf{while}(\mbox{true})\{ \vx := A \vx \}$ with the
expression $(R^{-1}\mathcal{M}R)\otimes X$, where $X$ is a
convex polyhedron describing the initial states.

\section{Abstract acceleration of loops with guards}
\label{sec:guard}

In this section, we consider loops of the form
$$
\mathsf{while}(G\vx\leq 0)\{ \vx  := A \vx \}
$$
modeled by the semantic function $G\srightarrow A$, as explained in
\S\ref{sec:preliminaries}. Given an initial set of states $X$, we
 compute an over-approximation of $Y=(G\srightarrow A)^*(X)$ using a
convex polyhedral domain, which after unfolding is expressed as
\vspace*{-0.4ex}
\begin{equation}\label{eq:GYdef}
  Y = X \cup \bigcup_{n\geq 1} \biggl(\Bigl(\bigwedge_{0\leq k\leq
    n-1} GA^k\Bigr)\rightarrow A^n\biggr)(X)
\vspace*{-0.4ex}
\end{equation}
The unfolding effectively computes the pre-condition of the guard $G$
on the initial state $X$ as $G A^k$.

\subsection{The simple technique}
\label{sec:guard:simple}

The expression $(G\srightarrow A)^*$ unfolded in
Eqn.~(\ref{eq:GYdef}) is too complex to be accelerated precisely.
A simple technique to approximate it safely is to exploit the following inclusion:
\begin{proposition}\label{prop:guard_inclusion}
  For any set $X$ and linear transformation $G\srightarrow A$,
  \begin{equation}
    \label{eq:guard_inclusion}
    (G\srightarrow A)^* \subseteq
    \mathit{id}\cup \Big( \underset{\footnotesize\begin{array}{c}\text{finally,
          take into}\\\text{account the guard}\end{array}}{\underbrace{(G\srightarrow A)}}\scirc~(G~\srightarrow \underset{\footnotesize\begin{array}{c}\text{acceleration}\\\text{without guard}\end{array}}{\underbrace{A^*}}) \Big)
  \end{equation}
\end{proposition}
\begin{figure}[t]
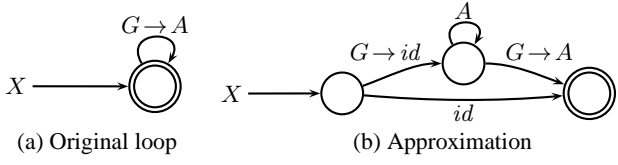

  \centering

  \psset{arrows=->,labelsep=0.2ex}
  \vspace*{4ex}

  \begin{tabular}{cc}
    \parbox{0.29\linewidth}{
      \begin{psmatrix}[mnode=r,colsep=4em]
	\rnode{start}{$X$} &
	\circlenode[doubleline=true]{A}{\hspace*{1em}}
      \end{psmatrix}
      \ncline[nodesepA=0.5ex]{start}{A}
      \nccurve[angleA=120,angleB=60,ncurv=3]{A}{A}\naput{$G\srightarrow A$}
    }
    &
    \parbox{0.7\linewidth}{
      \begin{psmatrix}[mnode=r,rowsep=-2ex,colsep=3em]
	&& \circlenode{I}{\hspace*{1em}} \\
	\rnode{pstart}{$X$} &
	\circlenode{start}{\hspace*{1em}} &&
	\circlenode[doubleline=true]{A}{\hspace*{1em}}
      \end{psmatrix}
      \ncline[nodesepA=0.5ex]{pstart}{start}
      \ncarc[arcangle=-5]{start}{A}\nbput{$\mathit{id}$}
      \nccurve[angleA=20,angleB=180]{start}{I}\naput{$G\srightarrow \mathit{id}$}
      \nccurve[angleA=0,angleB=160]{I}{A}\naput{$G\srightarrow A$}
      \nccurve[angleA=120,angleB=60,ncurv=3]{I}{I}\naput{$A$}
      \vspace*{1ex}

    } \\
    (a) Original loop &
    (b) Approximation
  \end{tabular}
  \caption{Original loop and approximation in term of
    invariant in accepting (double-lined) location, illustrating Prop.~\ref{prop:guard_inclusion}.}
  \label{fig:loopinv}
\end{figure}
Fig.~\ref{fig:loopinv} illustrates graphically
Prop.~\ref{prop:guard_inclusion}: the invariant attached to the
accepting location of Fig.~\ref{fig:loopinv}(a) is included in the
invariant attached to the accepting location of
Fig.~\ref{fig:loopinv}(b).
It is interesting to point out the fact that the abstract acceleration
techniques described in \cite{GH06,SJ12b} make
assumptions on the matrix $A$ and exploit convexity arguments so that the
inclusion~(\ref{eq:guard_inclusion}) becomes an equality. The idea
behind Prop.~\ref{prop:guard_inclusion} is applied to matrix
abstract domains to yield Prop.~\ref{prop:guard_overapprox}:

\begin{proposition}\label{prop:guard_overapprox}
  Let $A=R^{-1}JR$ with $J$ a real Jordan normal form, $T$ a
  template expression matrix, $\scrM=R^{-1}\alpha_T(J^*)R$ and $X$
  a convex polyhedron, then $(G\srightarrow A)^*(X)$ can be
  approximated by the convex polyhedron
  \begin{equation}\label{eq:guard_overapprox}
    X \sqcup (G\srightarrow A)\bigl(\scrM\otimes(X\sqcap G)\bigr)
  \end{equation}
\end{proposition}
This approach essentially consists of partially unfolding the loop, as
illustrated by Fig.~\ref{fig:loopinv}(b), and reusing abstract
acceleration without guard.  However, since the guard is only taken
into account \emph{after} the actual acceleration, precision is lost
regarding the variables that are not constrained by the
guard. Ex.~\ref{ex:running3} and Figures~\ref{fig:running2} and
\ref{fig:running3} in \S\ref{sec:overview:guard} illustrate this
weakness: $y$ is constrained by the guard, whereas $x$ remains unbounded.

\subsection{Computing and exploiting bounds on the number of
  iterations.}
\label{sec:guard:boundn}

To overcome the above issue, we propose a solution based on
finding the maximal number of iterations $N$ of the loop for any
initial state in $X$, and then to abstract the set of matrices $\{A^0,
\ldots, A^N\}$ instead of the set $A^*$.
The basic idea is that if there exists $N\sgeq 0$ such that $A^N X \cap
G = \emptyset$, then $N$ is an upper bound on the number of
iterations of the loop for any initial state in $X$.
Bounding the number of iterations is a classical problem in
termination analysis, to which our general approach provides a new,
original solution.

The following theorem formalizes this idea:
We assume now a guarded linear transformation $G\srightarrow J$
where $J$ is already a Jordan normal form.
\begin{theorem}\label{thm:MGX}
  Given a set of states $X$, a template expression matrix $T$ and the set of matrix
  $\scrM=\alpha_T(J^*)$, we define
  \begin{align*}
    \scrG &= \scrM\cap\{ M \;|\; \exists
    \vec{x}\in (X\cap G) \;:\; GM\vec{x}\leq 0 \} \\
    N&=\min \{ n \;|\; J^n\not\in \scrG \} 
  \end{align*}
  with the convention $\min \emptyset=\infty$. \\
  $\scrG$ is the set of matrices $M\subseteq\scrM$ of which the image
  of at least one input state $\vec{x}\in X$ satisfies $G$.
  If $N$ is bounded, then
  $$\textstyle
  (G\srightarrow J)^*(X)=\bigcup\limits_{n=0}^N (G\srightarrow J)^n(X)
  $$
\end{theorem}
\begin{proof}
  As $J^N\in\scrM\setminus\scrG$, we have $\forall
  \vec{x}\in (X\cap G)\;:\; GJ^N\vec{x}>0$, or in other words $(J^NX)\cap
  G=\emptyset$.  This implies that $\big((G\srightarrow J)^NX\big)\cap
  G=\emptyset$ and $(G\srightarrow J)^{N+1}X=\emptyset$.
\end{proof}
The definition of $\scrG$ and $N$ in Thm.~\ref{thm:MGX} can be transposed in
the space of vectors using the matrix shape $\Psi$:
\begin{theorem}\label{thm:PGX}
  Under the assumption of Thm.~\ref{thm:MGX}, and
  considering $\vec{m}(n)=\Psi^{-1}(J^n)$,  we have
  \vspace*{-0.4ex}
  \begin{align}
    \Psi^{-1}(\scrG) =& \left\{\begin{array}{l}\Psi^{-1}(\scrM)~\cap \\
                        \{ \vm \;|\; \exists \vec{x}: \vec{x}\in X\cap G :
    G\Psi(\vm)\vec{x}\leq 0 \}\end{array}\right.  \label{eq:PGX} \\
    N =& \min \{ n \;|\; \vec{m}(n) \not\in \Psi^{-1}(\scrG) \} \label{eq:NGX}
  \end{align}
\end{theorem}
Our approach to take into the guard is thus to compute a finite
bound $N$ with Eqns.~(\ref{eq:PGX}) and (\ref{eq:NGX}) and to replace in
Eqn.~(\ref{eq:guard_overapprox})
\vspace*{-0.4ex}
$$
\scrM=R^{-1}\cdot\alpha_T(J^*)\cdot R
\vspace*{-0.4ex}
$$
with
$\scrM'=R^{-1}\cdot\alpha_T(\{ J^n \;|\;
0\sleq n\sleq N\ssub 1\})\cdot R$
and $R$ the basis transformation matrix (see Prop.~\ref{prop:guard_overapprox}).

\begin{example}\label{ex:running5a}
  In our Examples~\ref{ex:running1}-\ref{ex:running4}, we had
  $\vec{m}(n)=\Psi^{-1}(J^n) = \vecb{1.5^n}{n}$,
  $\Psi^{-1}(\scrM)=(m_1\sgeq 1 \wedge m_2\sgeq 0 \wedge m_1\sadd
  m_2\sgeq 1 \wedge m_1\ssub m_2\sgeq 0.25)$, see
  Fig.~\ref{fig:running1} (light gray), $X=(x\in[1,3]\wedge
  y\in[0,2])$, see Fig.~\ref{fig:running3} (dark gray), and
  $G=(y\leq 3)$. The second term of the
  intersection in Eqn.~(\ref{eq:PGX}) evaluates to $m_2\sleq 3$, thus
  $\Psi^{-1}(\scrG)=(m_1\sgeq 1 \wedge 0\sleq m_2\sleq 3 \wedge
  m_1\sadd m_2\sgeq 1 \wedge m_1\ssub m_2\sgeq 0.25)$. This gives us
  through Eqn.~(\ref{eq:NGX}) the bound $N=4$ on the number of loop
  iterations.  The abstraction $\scrM'=\alpha_T(\{ J^n \;|\; 0\sleq
  n\sleq 3\})$, depicted on Fig.~\ref{fig:running1} (dark gray),
  removes those matrices from $\scrM$ that do not contribute to the
  acceleration result due to the guard.  Finally, we accelerate using
  $\scrM'$ and obtain the result shown in Fig.~\ref{fig:running3}
  (medium gray).
\end{example}

\noindent %
More details about these computations are given in the next section.

\subsection{Technical issues}
\label{sec:guard:boundnt}
Applying Thms.~\ref{thm:MGX}~and~\ref{thm:PGX} requires
many steps.  First, we  approximate
$\Psi^{-1}(\scrG)$ (see Eqn.~(\ref{eq:PGX})), and then as a second
step we can approximate the maximum number of iterations $N$ according
to Eqn.~(\ref{eq:NGX}).  Finally, we have to compute
$\alpha_T\{ J^n\;|\; 0\sleq n\sleq N\ssub 1\}$.

\paragraph{Approximating $\Psi^{-1}(\scrG)$.}
Let us denote $Q=\Psi^{-1}(\scrG)$ and $P=\Psi^{-1}(\scrM)$.  $Q$ is
defined in Eqn.~(\ref{eq:PGX}) by the conjunction of quadratic
constraints on $\vec{x}$ and $\vm$ followed by an elimination of
$\vec{x}$. Exact solutions exist for this problem, but they are
costly.
The alternative adopted in this paper is to approximate $Q$ by
quantifying $\vec{x}$ on the bounding box of $X\cap G$ instead of
quantifying it on $X\cap G$. Let us denote the bounding box of a
polyhedron $Z$ with the vector of intervals
$[\underline{Z},\overline{Z}]$. We have
\begin{equation}\label{eq:P'}
  \begin{aligned}
    Q \;\subseteq\; &
    P\cap \{ \vm \;|\; \exists \vec{x}:
    \vec{x}\in [\underline{X\cap G},\overline{X\cap G}] \wedge
    G\Psi(\vm)\vec{x}\leq 0 \} \\[-1ex]
    &= \\[-1.5ex]
    &P\cap \{ \vm \;|\;
  G\Psi(\vm)[\underline{X\cap G},\overline{X\cap G}]\leq 0
  \} \;\defeq\; Q'
\end{aligned}
\end{equation}
$Q'$ is defined by intersecting $P$ with
\emph{interval-linear constraints} on $\vm$ and it is much easier to compute
than $Q$: one can use
\begin{compactitem}
\item algorithms for interval linear constraints
  \cite{CMWC10}, in particular interval linear programming \cite{Rohn06}, or
\item the linearization techniques of
  \cite{Mine06} that are effective if the vectors $\vm$
  are ``well-constrained'' by $P$.
\end{compactitem}
In this paper, we exploit the last method which is implemented in the APRON
library \cite{jeannetmine09}.
\begin{example}\label{ex:running5}
  Coming back to our running
  Examples~\ref{ex:running1}-\ref{ex:running4} and \ref{ex:running5a}, we have
  \vspace*{-2ex}
  \begin{multline*}
   \hspace*{-1em}  G\Psi(\vec{m})[\underline{X\cap G},\overline{X\cap G}] = (0~1~-3)
    \begin{pmatrix}
      m_1 & 0 & 0 \\
      0 & 1 & m_2 \\
      0 & 0 & 1
    \end{pmatrix}
    \vect{[1,3]}{[0,2]}{1} \\
    =
    (0~1~-3)
    \vect{[1,3]m_1}{[0,2]+m_2}{1}
    =
    m_2+[-3,-1]
  \end{multline*}
  Hence $Q'=P\cap\denotation{m_2+[-3,-1]\sleq 0}=P\cap\denotation{m_2\sleq 3}$
\end{example}

\paragraph{Approximating the maximum number of iterations $N$.}
Computing $N$ as defined in Eqn.~(\ref{eq:NGX}) is not easy
either, because the components of vector $\vec{m}(n)$ are functions of
defined by Eqn.~(\ref{eq:varphi}). Our approach is to exploit a
matrix of template expressions.
\begin{proposition}\label{prop:Ntemplate}
  Under the assumption of Thm.~\ref{thm:PGX}, for any polyhedron
  $Q'\supseteq Q$ and template expression matrix $T'$,
  \begin{equation}
    N\leq \min
    \Bigl\{ n \;|\; \exists j: T'_j\cdot \vec{m}(n) >
    \sup_{\mathclap{\vm\in Q'}} T'_j\cdot \vm \Bigr\}
    \label{eq:minNumberIt}
  \end{equation}
\end{proposition}
\begin{proof}
  We have $Q\subseteq Q'\subseteq \alpha_{T'}(Q')$. From
  $Q\subseteq \alpha_{T'}(Q')$, \\[0.5ex]
  it follows $\vec{m}(n)\not\in\alpha_{T'}(Q') \Rightarrow
  \vec{m}(n)\not\in Q$, \\[0.5ex]
  and $\min\{ n \;|\;
    \vec{m}(n)\not\in\alpha_{T'}(Q')\} \geq \min\{n\;|\; \vec{m}(n)\not\in
    Q\}=N$. \\[0.5ex]
  $\vec{m}(n)\not\in\alpha_{T'}(Q')$ is equivalent to
  $\exists j:
  T'_j\cdot \vec{m}(n)>\;\sup\limits_{\mathclap{\vm\in \alpha_{T'}(Q')}}\;
  T'_j\cdot \vm$ and $\displaystyle \sup_{\mathclap{\vm\in Q'}}
  T'_j\cdot \vm=\sup_{\mathclap{\vm\in\alpha_{T'}(Q')}} T'_j\cdot
  \vm$, hence we get
  the result.
\end{proof}
In our implementation we compute such a $Q'$ as described in the
previous paragraph and we choose for $T'$ the template matrix $T$
considered in \S\ref{sec:accelerating:bound}, to
which we may add the constraints of $Q'$.

The computation of such minima
ultimately relies on the Newton-Raphson method for solving
equations. Our implementation deals with the same cases as those
mentioned in \S\ref{sec:accelerating:bound} and in other cases safely
returns $N_j=\infty$. We refer to \pponly{the extended version \cite{JSS13}}\rronly{appendix \ref{sec:maxiter}} for details.
\begin{example}
  Coming back to Examples~\ref{ex:running1}-\ref{ex:running4} and \ref{ex:running5a}-\ref{ex:running5}, and considering
  $T_j=(0~ 1)$, we have \\
  $T_j\scdot\vec{m}(n)=(0~1)\scdot\vecb{1.5^n}{n}=n$,
    $e_j=\sup\limits_{\mathclap{\vm\in Q'}} T_j\scdot\vec{m}=\sup\limits_{\mathclap{\vm\in Q'}} m_2 = 3$
    and $N_j=\min\{ n \;|\; T_j\cdot\vec{m}(n) > e_j \} = 4$, \\
    which proves that $4$ is an upper bound on the maximum number of
  iterations of the loop, as claimed in Ex.~\ref{ex:running4}.
\end{example}

\paragraph{Computing $\alpha_T\{ J^n\;|\; 0\sleq n\sleq N\ssub 1\}$.}

If no finite upper bound $N$ is obtained with
Prop.~\ref{prop:Ntemplate} (given an input polyhedron $X$), then
we apply the method of \S\ref{sec:guard:simple}. Otherwise, we
replace in Eqn.~(\ref{eq:guard_overapprox}) the set
$\scrM=\alpha_T(A^*)$ with $\scrM'=\alpha_T(\{ J^n \;|\; 0\sleq
n\sleq N\ssub 1\})$.  This set $\scrM'$ can be computed using the
same technique as those mentioned in
\S\ref{sec:accelerating:bound} and detailed in \pponly{\cite{JSS13}}\rronly{\ref{sec:bound}}, or
even by enumeration if $N$ is small.
Ex.~\ref{ex:running4} and Figs.~\ref{fig:running1}-\ref{fig:running3} illustrate
the invariant we obtain this way on our running example.
\begin{example}[Running example with more complex guard]\label{ex:running7}
  Coming back to Examples~\ref{ex:running1}-\ref{ex:running4}, we
  consider the same loop but with the guard $x\sadd 2y\sleq 100$
  represented by the matrix $(1\; 2\; -\! 100)$.
  Compared to Ex.~\ref{ex:running5} we have now
  $G\Psi(\vec{m})[\underline{X\cap G},\overline{X\cap G}]=[1,3]m_1
  \sadd [0,4]\sadd 2m_2 \ssub 100$, hence, if $P$ is defined by
  $\varphi_\scrM$ as in Ex.~\ref{ex:running1},
  \vspace*{-0.4ex}
  \begin{align*}
    Q' &= P\cap \denotation{[1,3]m_1\sadd 2m_2\sleq [96,100]} \\
    &\subseteq P\cap\denotation{m_1\sadd 2\sadd 2m_2\sleq 100}\\
    & \quad\text{linearization based on } P\Rightarrow m_1\sgeq
    1\Rightarrow m_1\sadd 2\sleq[1,3]m_1
    \vspace{0.4ex}
  \end{align*}
  Using $Q'$ to bound the number of iterations according to
  Eqn.~(\ref{eq:minNumberIt}) leads to $N=11$ (obtained with the
  template expression $T_j=(1~2)$ tasken from the guard). At last
  we obtain the invariant $Z$ depicted in Fig.~\ref{fig:running4}. If
  instead of octagonal template expressions with $\ell=1$ in
  Eqn.~(\ref{eq:logahedratemplate}), we choose $\ell=3$, we do not
  improve the bound $N=11$ but we still obtain the better
  invariant $Z'$ on Fig.~\ref{fig:running4}.
\end{example}

\begin{figure}[t]
  \psset{xunit=0.16em,yunit=0.5em,griddots=0,subgriddiv=0,gridlabels=3pt}

  \vspace*{1ex}
  \begin{pspicture}(-5,-3)(170,15)
    \pspolygon[fillstyle=solid,fillcolor=LightGrey](1,0)(1,2)(1.5,3.75)(15.38,13)(114,13)(150,1)(3,0)
    \pspolygon[fillstyle=solid,fillcolor=DarkGrey](1,0)(1,2)(1.688,3.375)(2.81,4.6)(4.594,5.812)(8.4,7.336)(16.46,9.127)(57.13,13)(114,13)(137,5.217)(3,0)
      \psframe[fillstyle=solid,fillcolor=darkgray](1,0)(3,2)
      {\scriptsize
      \psaxes[ticksize=2pt,tickstyle=full,Dx=20,Dy=2]{->}(0,0)(0,0)(160,15)}
      \rput(160,1){{\small $x$}}
      \rput(3,14.45){{\small $y$}}
      \psline[linestyle=dashed](70,15)(100,0)
      \rput(50,14.5){\small $x\sadd 2y\sleq 100$}
      \rput(110,9){$Z'$}
      \rput(130,3){$Z$}
    \end{pspicture}
    \caption{Initial set of states $X$ (dark gray), loop invariants
      $Z$ (light gray) and $Z'$ (medium gray) obtained
      in Ex.~\ref{ex:running7}.
    }
    \label{fig:running4}
\end{figure}
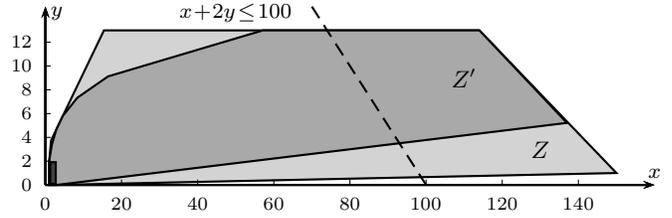


\subsection{Summary}
\label{sec:guard:summary}

We summarize now our method in the general case:

Given a guarded linear transformation $G\srightarrow A$, with
$J=RAR^{-1}$ the (real) Jordan normal form of $A$, its associated
matrix shape $\Psi$, a template expression matrix $T$ and a convex
polyhedron $X$ representing a set of states, we compute an
overapproximation $(G\srightarrow A)^*(X)$ with
\vspace*{-0.4ex}
$$
 X \;\sqcup\; (G\srightarrow
  A)\bigl((R^{-1}\scrM) \otimes Y\bigr)
\vspace*{-0.4ex}
$$
where
\begin{tabular}[t]{ll}
1. & $Y=R(X\sqcap G)$, \\
2. & $P=\alpha_T(\Psi^{-1}(J^*))$, \\
3. & $Q'=P \sqcap \{
  GR^{-1}\Psi(\vm)[\underline{Y},\overline{Y}]\leq 0 \}$, \\[0.5ex]
4. & $N$ approximated by some $N'$ using Prop.~\ref{prop:Ntemplate}, \\
5. & $\scrM=
  \begin{cases}
    \alpha_T(J^*) & \text{if } N'=\infty \\
    \alpha_T(\{ J^n \;|\; 0\sleq n\sleq N'\ssub 1 \}) & \text{otherwise}
  \end{cases}
  $
\end{tabular}
\smallskip


\section{Implementation and Experiments}
\label{sec:exp}

We implemented the presented approach in a prototype tool, evaluated
over a series of benchmark examples with various kinds of linear
loops, possibly embedded into outer loops, and compared it to
state-of-the-art invariant generators.

\begin{table*}[t]
\centering
\footnotesize
\begin{tabular}{|l|*{4}{@{\;}c@{\;}}|*{3}{r@{\;\;}}|*{3}{r@{}}@{\;\;\;\;\;}r@{}r@{}@{\;}|*{3}{r@{}l@{\;\;\;}}|*{2}{r@{}l@{\;\;\;}}|}
\hline
& \multicolumn{4}{c|}{characteristics}
& \multicolumn{8}{c|}{inferred bounds}
& \multicolumn{10}{c|}{analysis time (sec)} \\
name            & type                     & $s_\mathit{max}$ &
\#var & \#bds
& IProc & Sti & J & \multicolumn{3}{c}{J vs IProc} & \multicolumn{2}{c|}{J vs Sti} & \multicolumn{2}{c}{IProc}
& \multicolumn{2}{c}{Sti} & \multicolumn{2}{c}{J} &\multicolumn{2}{|c}{ JSage} & \multicolumn{2}{c|}{JAna} \\ \hline
\multicolumn{23}{|l|}{Examples with single loops} \\ \hline
parabola\_i1& $\neg s$,$\neg c$,$g$ & 3 & 3 & 60
& 46 & 38 & 54 & +8,& +17 && +16,&+12 & 0&.007 & 237& & 2&.509 & 2&.494 & 0&.015  \\
parabola\_i2& $\neg s$,$\neg c$,$g$ & 3 & 3 & 60
& 36 & 32 & 54 & +18,& +6 && +22,&+13 & 0&.008 & 289& & 2&.509 & 2&.494 & 0&.015 \\
cubic\_i1& $\neg s$,$\neg c$,$g$      & 4 & 4 & 80
& 54 & 34 & 72 & +18,& +26 & & +38,&+12 & 0&.015 & 704& & 2&.474 & 2&.393 & 0&.081 \\
cubic\_i2& $\neg s$,$\neg c$,$g$      & 4 & 4 & 80
& 34 & 28 & 72 & +38,& +9&, -2 & +44,&+11 & 0&.018 & 699& & 2&.487 & 2&.393 & 0&.094 \\
exp\_div& $\neg s$,$\neg c$,$g$      & 2 & 2 & 32
& 24 & 21 & 28 & +4,& +6&, -2 & +7,&+7 & 0&.004 & 31&.6 & 2&.308 & 2&.299 & 0&.009 \\
oscillator\_i0& $s$,$c$,$\neg g$ & 1 & 2 & 28
& 1 & 0 & 24 & +23,& +0&, -1 & +24,&+0 & 0&.004 & 0&.99 & 2&.532 & 2&.519 & 0&.013 \\
oscillator\_i1& $s$,$c$,$\neg g$ & 1 & 2 & 28
& 0 & 0 & 24 & +24,& +0 && +24,&+0 & 0&.004 & 1&.06 & 2&.532 & 2&.519 & 0&.013 \\
inv\_pendulum& $s$,$c$,$\neg g$  & 1 & 4 & 8
& 0 & 0 & 8 & +8,&+0 & & +8,&+0 & 0&.009 & 0&.920 & 65&.78 & 65&.24 & 0&.542 \\
convoyCar2\_i0 & $s$,$c$,$\neg g$ & 2 & 5 & 20
& 3 & 4 & 9 & +6,&+0 && +5,&+1 & 0&.007 & 0&.160 & 5&.461 & 4&.685 & 0&.776 \\
convoyCar3\_i0 & $s$,$c$,$\neg g$ & 3 & 8 & 32
& 3 & 3 & 15 & +12,&+0 && +12,&+0 & 0&.010 & 0&.235 & 24&.62 & 11&.98 & 12&.64 \\
convoyCar3\_i1 & $s$,$c$,$\neg g$ & 3 & 8 & 32
& 3 & 3 & 15 & +12,&+0 && +12,&+0 & 0&.024 & 0&.237 & 23&.92 & 11&.98 & 11&.94 \\
convoyCar3\_i2 & $s$,$c$,$\neg g$ & 3 & 8 & 32
& 3 & 3 & 15 & +12,&+0 && +12,&+0 & 0&.663 & 0&.271 & 1717& & 11&.98 & 1705& \\
convoyCar3\_i3 & $s$,$c$,$\neg g$ & 3 & 8 & 32
& 3 & 3 & 15 & +12,&+0 && +12,&+0 & 0&.122 & 0&.283 & 1569& & 11&.98 & 1557& \\ \hline
\multicolumn{23}{|l|}{Examples with nested loops} \\ \hline
thermostat& $s$,$\neg c$,$g$          & 2 & 3 & 24
& 18 & 24 & 24 & +6,& +1,& -1 & +0,&+6 & 0&.004 & 0&.404 & 4&.391 & 4&.379 & 0&.012 \\
oscillator2\_16& $\neg s$,$c$,$g$ & 1 & 2 & 48
& 9 & t.o. & 48 & +27,& +0 && \multicolumn{2}{@{}c@{}|}{t.o.} & 0&.003 & t.&o. & 4&.622 & 4&.478 & 0&.144 \\
oscillator2\_32& $\neg s$,$c$,$g$ & 1 & 2 & 48
& 9 & t.o. & 48 & +39,& +0 && \multicolumn{2}{@{}c@{}|}{t.o.} & 0&.003 & t.&o. & 4&.855 & 4&.490 & 0&.365 \\ \hline
\end{tabular}
 \\[0.7ex]
\textbf{``$s/\neg s$''}: stable/unstable loop; \textbf{``$c/\neg c$''}: has complex eigenvalues or not; \textbf{``$g/\neg g$''}: loops with/without
guard; \textbf{$s_\mathit{max}$}: size of largest Jordan block(s);
\textbf{``\#var''}: nb. of variables; \textbf{``\#bds''}: nb. of
bounds to be inferred at all control points;
\textbf{``IProc''},\textbf{``Sti''},\textbf{``J''}: nb.
of finite bounds inferred by
\textsc{Interproc}, \textsc{Sting}, and our method (``J'');
\textbf{``J vs IProc''},\textbf{``J vs Sti''}: \textbf{``+x,+y,[-z]''}:
nb. of infinite bounds becoming finite, nb. of improved finite
bounds[, nb. of less precise finite bounds (omitted if 0)] obtained with our
method over \textsc{Interproc}, \textsc{Sting}; \textbf{``analysis
  time''}: running times (seconds), with ``\textbf{JSage}'',''\textbf{JAna}'' corresponding to
computation of the Jordan normal form using \textsc{Sage} and the
analysis itself, and \textbf{``J''} being the sum of these two times. ``t.o.'' means time out after one hour.
\caption{Experimental results. }
\label{tab:exp}
\end{table*}

\subsection{Implementation}
\label{sec:impl}

We integrated our method in an abstract interpeter based on the
\textsc{Apron} library \cite{jeannetmine09}. We detail below some
of the issues involved.

\paragraph{Computation of the Jordan normal form}
To ensure soundness, we have taken a symbolic approach using
the computer algebra software
\textsc{Sage}\footnote{\url{www.sagemath.org}} for computing
the Jordan normal forms and transformation matrices over the
  field of algebraic numbers. The matrices are then
approximated by safe enclosing interval matrices.

\paragraph{Loops with conditionals}
Loops of which the body contains conditionals like $\textit{while}(G)
\{ \textit{ if }(C)\; \vx' \seq A_1\vx \textit{ else } \vx' \seq A_2\vx \}$ can
be transformed into two self-loops $\tau_1, \tau_2$ around a ``head''
location that are executed in non-deterministic order:
\begin{figure}[h]
\centering
     \psset{labelsep=5.5em,arrows=->}
	\circlenode{A}{head}
      \nccurve[angleA=30,angleB=-30,ncurv=3]{A}{A}\nbput{$G\wedge \neg
        C\srightarrow A_2$}
      \nccurve[angleA=150,angleB=-150,ncurv=3]{A}{A}\naput{$G\wedge C\srightarrow A_1$}
\end{figure}
We iteratively accelerate each self-loop separately $(\tau_1^\otimes
\circ \tau_2^\otimes)^*$. Since convergence of the outer loop
$(\cdot)^*$ is not guaranteed in general, we might have to use
widening. Yet, practical experience shows that in many cases a fixed
point is reached after a few iterations.

\paragraph{Nested loops}
We could use a similar trick to transform nested loops
$\textit{while}(G_1) \{ \textit{while}(G_2) \{\vx' \seq A_2\vx
\};~\vx' \seq A_1\vx \}$ into multiple linear self-loops $\tau_1, \tau_2, \tau_3$
by adding a variable $y$ (initialized to $0$) to encode the control flow:
$$
\begin{array}{ll}
\tau_1: & (G_1\vx\sleq 0 \wedge y\seq 0) \rightarrow (\vx'\seq \vx \wedge y'\seq 1)\\
\tau_2: & (G_2\vx\sleq 0 \wedge y\seq 1) \rightarrow (\vx'\seq A_2\vx \wedge y'\seq 1)\\
\tau_3: & (G_2\vx\sgt 0 \wedge y\seq 1) \rightarrow (\vx'\seq A_1\vx \wedge y'\seq 0)
\end{array}
$$ However, the encoding of the control flow in an integer variable is
ineffective because of the convex approximation of the polyhedral
abstract domain, this transformation causes an inacceptable loss of
precision.
For this reason, we accelerate only inner loops in nested loops
situations.  Our experimental comparison shows that computing precise
over-approximations of inner loops greatly improves the analysis of
nested loops, even if widening is applied to outer loops.

\subsection{Evaluation}
\label{sec:exp-eval}

\paragraph{Benchmarks}
Our benchmarks listed in Table~\ref{tab:exp} include various examples
of linear loops as commonly found in control software.  They
contain filters and integrators that correspond to various cases of
linear transformations (real or complex eigenvalues, size of Jordan
blocks).
\emph{parabola} and \emph{cubic} are loops with polynomial behavior
(similar to Fig.~\ref{fig:parabola-code}), \emph{exp\_div} is
Ex.~\ref{ex:running4} and \emph{thermostat} is
Ex.~\ref{Ex:iir-filter-example}.  \emph{inv\_pendulum} is the
classical model of a pendulum
balanced in upright position by movements of the cart it is mounted
on.
\emph{oscillator} is a damped oscillator, and
\emph{oscillator2} models a pendulum that is only damped in a
range around its lowest position, as if it was grazing the ground, for
example. This is modeled using several modes.
In \emph{ConvoyCar} \cite{SSM04},
a leading car is followed by one or more cars, trying to
maintain their position at 50m from each other:
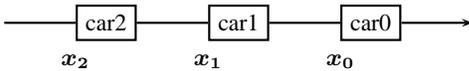
\begin{center}
  \boldmath

  \begin{psmatrix}[rowsep=3ex,colsep=3em,emnode=p,mnode=r]
    & \psframebox{car2} & \psframebox{car1} & \psframebox{car0} &
    \\
    &&&&
  \end{psmatrix}
  \ncline{1,1}{1,2}
  \ncline{1,2}{1,3}
  \ncline{1,3}{1,4}
  \ncline[arrows=->]{1,4}{1,5}
  \nput{90}{2,2}{$x_2$}
  \nput{90}{2,3}{$x_1$}
  \nput{90}{2,4}{$x_0$}
\end{center}\vspace*{-2ex}
The equations for $N-1$ following cars are:
$$
\text{for all }i\in[1,N\ssub 1]: \ddot{x}_i=c(\dot{x}_{i-1}-\dot{x}_{i})+d(x_{i-1}-x_{i}-50)
$$
We analyzed a discretized version of this example to show that there is no collision, and
to compute bounds on the relative positions and speeds of the
cars. This example is particularily interesting because the real Jordan form of the
loop body has blocks associated to complex eigenvalues of size $N-1$.

All benchmarks have non-deterministic initial states (typically
bounding boxes). Some benchmarks
where analyzed for different sets of initial states (indicated by suffix \texttt{\_i}$x$).

\paragraph{Comparison}
Existing tools can only handle subsets of these
examples with reasonable precision.
We compared our method with
\begin{compactitem}
\item \textsc{Interproc} \cite{interproc} that implements standard
  polyhedral analysis with widening \cite{CH78};
\item \textsc{Sting} that implements the method of \cite{CSS03,SSM04}.
\end{compactitem}
A comparison with the \textsc{Astr\'ee} tool on the
\textit{thermostat} example has been given in Section~\ref{sec:summary-exp}.
A detailed \emph{qualitative} comparison between various methods described in
\S\ref{sec:rw} is shown in Table~\ref{tab:related}.

\begin{table*}[t]
  \centering
  \begin{tabular}{|l|l|c@{$\;$}c@{$\;$}c@{$\;$}c@{$\;$}c@{$\;$}c|}
    \hline
   & & & linear &  & relational  &  & \\
  type of linear  & illustrating & & abstr. accel. & ellipsoids & abstraction & \textsc{Aligator} & \textsc{Sting} \\
  transformation & examples & this paper & \cite{GH06,ACI10} & \cite{Feret05b,Monniaux05,RJGF12} &
    \cite{LPY01,Tiw03,Monniaux09,ST11} & \cite{Kovacs08b} & \cite{SSM04,CSS03} \\
   \hline
    \parbox{9em}{
      translations/resets \\
      $\lambda_i\seq 1 \wedge s_k\sleq 2 \vee \lambda_i\seq 0$
    } &
    & yes & yes & no & yes & yes & yes \\[1.5ex] \hline
    \parbox{9em}{
      polynomials \\
      $\lambda_i\in\{0,1\}$
    } &
   \parbox{9em}{\small\sf parabola, cubic}
   & yes & no & no & no & yes & no \\[1.5ex] \hline
    \parbox{9em}{
      exponentials \\
      $\lambda_i\in\mathbb{R}^+ \wedge s_k\seq 1$
    } &
    \parbox{9em}{\small\sf exp\_div} &
    yes & no & no & if $|\lambda_i|\slt 1$ & no & if $|\lambda_i|\slt 1$\\[1.5ex] \hline
    \parbox{9em}{
      rotations \\
      $\lambda_i\in\mathbb{C} \wedge s_k\seq 1$
    } &
    \parbox{9em}{\small\sf oscillator, oscillator2, inv\_pendulum}
    & yes & no & if $|\lambda_i|\slt 1$ & no & no & no \\[1.5ex] \hline
    \parbox{9em}{
      non-diagonalizable \& \\ non-polynomial \\
      $s_k\sgeq 2 \wedge \lambda_i\notin\{0,1\}$
    } &
    \parbox{9em}{\small\sf thermostat, convoyCar} &
    yes & no & no & no & no & if $|\lambda_i|\slt 1$ \\[2.75ex]\hline
    \parbox{9em}{
      unstable \&\\  non-polynomial \\
      $|\lambda_i|\sgt 1$
    } &
    \parbox{9em}{\small\sf exp\_div, oscillator2} &
    yes & no & no & no & no & no \\[0ex] \hline
    \parbox{9em}{
      loop guard handling
    } &
    \parbox{9em}{\small\sf parabola, cubic, exp\_div, thermostat,
    oscillator2} &
    yes & yes & partially & yes & no & yes \\[2.75ex] \hline
    \parbox{9em}{
      inputs/noise
    } & &
    no & yes & yes & yes & no & yes \\ \hline
   abstract domain & & polyhedra & polyhedra & ellipsoids & \parbox{5em}{\centering template \\ polyhedra}
   & \parbox{5em}{\centering polynomial \\ equalities} & polyhedra \\
    \hline
  \end{tabular}~\\[1ex]
{\small $s_k$ is the size of a Jordan block and $\lambda_i$ its associated eigenvalue.}
  \caption{Classification of linear loops and the capabilities of
    various analysis methods to infer precise invariants}
  \label{tab:related}
\end{table*}

\paragraph{Results}
Table~\ref{tab:exp} lists our experimental results.\footnote{A
  detailed account of the benchmarks and the obtained invariants can
  be found on \url{http://www.cs.ox.ac.uk/people/peter.schrammel} \url{/acceleration/jordan/}.}  We
compared the tools based on the number of \emph{finite} bounds
inferred for the program variables in each control point.  Where
applicable, we report the number of bounds more/less precise finite
bounds inferred by our tool in comparison to the other tools.

We note that our analysis dramatically improves the accuracy over the
two competing techniques. On all the benchmarks considered, it
generally provides strictly stronger invariants and is practically
able to infer finite variable bounds whenever they exist.
For instance, for the \textit{thermostat} example
(Fig.~\ref{fig:thermostat}), we infer that at least 6.28 and at most
9.76 seconds are continuously spent in heating mode and 10.72 to 12.79
seconds in cooling mode. \textsc{Interproc} just reports that time is
non-negative and \textsc{Sting} is able to find the much weaker bounds
$[0.88,13.0]$ and $[0.94,19.0]$.
On the \textit{convoyCar} examples, our method is the only one that is able
to obtain non-trivial bounds on distances, speeds and accelerations.

Yet, this comes at the price of increased computation times in
general. \textsc{Interproc} is significantly faster on all examples;
\textsc{Sting} is faster on half of the benchmarks and significantly
slower on the other half: for two examples, \textsc{Sting} does not
terminate within the given timeout of one hour, whereas our tool gives precise bounds
after a few seconds.  It must be noted that part of the higher
computation time of our tool is a ``one-time'' investment for
computing loop accelerations that can pay off for multiple
applications in the course of a deployment of our approach in a tool
such as \textsc{Astr\'ee}. In all but two of the examples (the
exceptions being \texttt{convoyCar3\_i[2|3]}), the one-time cost
dominates the overall cost of our analysis.  All in all, computation
times remain reasonable in the view of the tremendous gain in
precision.


\section{Related work}
\label{sec:rw}

\paragraph{Invariants of linear loops.}
The original \emph{abstract acceleration} technique of Gonnord et al
\cite{GH06,FG10}
precisely abstracts linear loops performing translations,
translations/resets and some other special cases.
The \emph{affine derivative closure} method \cite{ACI10}
approximates any loop by a
translation; hence it can handle any linear transformation,
but it is only precise for translations.
The tool \textsc{InvGen} \cite{GR09} uses template-based constraint
solving techniques for property-directed invariant generation, but it
is restricted to integer programs.

Methods from linear filter analysis target \emph{stable} linear
systems in the sense of Lyapunov stability.  These techniques consider
\emph{ellipsoidal domains}.  Unlike our method, they are able to
handle inputs, but they do not deal with guards.
Feret \cite{Feret04,Feret05b} designs specialized domains for
first and second order filters. These methods are implemented
in the \textsc{Astr\'ee} tool \cite{BCC+03}.
%
Monniaux \cite{Monniaux05} computes interval bounds for composed filters.
%
Roux et al \cite{RJGF12} present a method based on semidefinite programming
that infers both shape and ratio of an ellipsoid that is an invariant
of the system.
In contrast, our method does not impose any stability requirement.


Colon et al \cite{CSS03,SSM04} describe a method, implemented in the tool
\textsc{Sting}, for computing single
inductive, polyhedral invariants for loops based on non-linear
constraint solvers. It is a computationally expensive method and in
contrast, our approach is able to infer sound polyhedral
over-approximations for loops where the only \emph{inductive polyhedral}
invariant is \textsf{true}.

\emph{Relational abstraction} methods
\cite{LPY01,Tiw03,ST11} aim at finding a relation between
the initial state $\vx$ and any future state $\vx'$ in continuous
linear systems $d\vx(t)/dt=A\vx(t)$, which is a problem similar to the
acceleration of discrete linear loops.  The invariants are computed
based on off-the-shelf quantifier elimination methods over real
arithmetic. In contrast to our method, they handle only
diagonalizable matrices $A$.
\cite{Monniaux09} proposes a similar approach for discrete loops.
All these works compute a (template) polyhedral relation between
input and output states, and do not take into account the actual
input states. Hence, in contrast
to our method, they are unable to capture accurately unstable and rotating
behavior.

\emph{Strategy iteration} methods \cite{GGTZ07,GS07b,GSA+12} compute
best inductive invariants \emph{in the abstract domain} with the help
of mathematical programming.  They are not restricted to simple loops
and they are able to compute the invariant of a whole program at once.
However, they are restricted to template domains, \eg~ template
polyhedra and quadratic templates, and hence, unlike our method, they
are unable to \emph{infer} the shape of invariants.

The \textsc{Aligator} tool \cite{Kovacs08b} infers loop invariants that are polynomial \emph{equalities}
by \emph{solving the recurrence equations}
representing the loop body in closed form.  The class of programs that
can be handled by Aligator is incomparable to that considered in this
paper.  Whereas Aligator handles a subset of non-linear polynomial
assignments our work is currently restricted to linear assignments.
In contrast, we take into account linear inequality loop conditions, and can compute
inequality invariants.

\paragraph{Bounding loop iterations.}
Many papers have investigated the problem of bounding the number
of iterations of a loop, either
for computing widening thresholds \cite{SK06}
by linear extrapolation,
termination analysis \cite{ADFG10} using ranking functions or
for WCET analysis \cite{KKZ11} based on solving recurrence equations.
%
%
Yazarel and Pappas \cite{YP04} propose a method for computing time
intervals where a linear continuous system fulfills a given safety
property.  Their method can handle complex eigenvalues by performing
the analysis in the polar coordinate space. However, they can only
deal with diagonalizable systems.

All these methods assume certain restrictions on the form
of the (linear) loop, whereas our approach applies to any linear loop.

\section{Conclusion}
\label{sec:concl}

We presented a novel abstract acceleration method for discovering
polyhedral invariants of loops with general linear transformations.
It is based on abstracting the transformation matrices induced by any
number of iterations, polyhedral matrix multiplication, and a method
for estimating the number of loop iterations that also works in case
of exponential and oscillating behavior.  Our experiments show that we
are able to infer invariants that are out of the reach of existing
methods.  The precise analysis of linear loops is an essential feature
of static analyzers for control programs. Precise loop invariants are
equally important for alternative verification methods based on model
checking, for example. Further possible applications include
termination proofs and deriving complexity bounds of algorithms.

\paragraph{Ongoing work.}
In this paper, we considered only closed systems, \ie~without
inputs. However, important classes of programs that we want to
analyze, \eg~ digital filters, have inputs. Hence,
we are extending our methods to loops of the form
$G\vx+H\vec{\xi}\leq 0$ $\rightarrow\vx'=A\vx+B\vec{\xi}$ with inputs
$\vec{\xi}$ (similar to \cite{SJ12b}).
Moreover, our method easily generalizes to\rronly{ the analysis of}
linear continuous systems $d\vx(t)/dt=A\vx(t)$, \eg~representing the
modes of a hybrid automaton, by considering their analytical solution
$\vx(t)=e^{At}\vx(0)$.


\bibliographystyle{abbrvnat}
\label{sec:bibliography}
\bibliography{jabref,mybib}

\clearpage

\appendix

\section{Bounds on coefficients of powers of Jordan matrices}
\label{sec:bound}

The aim of this annex is to provide lower and upper bounds to expressions of
the form
\begin{gather}
  \phi(n)=\mu_1\varphi[\lambda_1,\theta_1,r_1,k_1](n) +
  s\cdot\mu_2\varphi[\lambda_2,\theta_2,r_2,k_2](n) \label{eq:bound:phin}
  \\
  \text{with }
  \varphi[\lambda,\theta,r,k](n)\defeq\tbinom{n}{k}\lambda^{n-k}
  \cos((n-k)\theta-r\tfrac{\pi}{2}) \\
  \text{and }
\left\{
\begin{array}{l}
\lambda_i\geq 0,\\
\theta_i\in[0,\pi],\\
r_i\in\{0,1\},\\
\mu_1>0, \mu_2\geq 0, \\
s\in\{-1,1\}
\end{array}
\right.\notag
\end{gather}
$\phi(n)$ is the linear
combination of the two different coefficients
$\varphi[\lambda_i,\theta_i,r_i,k_i](n)$ of the matrix $J^n$ where
$J$ is in real Jordan normal form, see
Section~\ref{sec:accelerating:jordan}.

When $\lambda_i>0$ for $i=1,2$,
by defining $\gamma_i=\log\lambda_i$
we extend $\phi$ to positive reals as follows:
\begin{gather}
  \phi(t)=\mu_1\varphi[\gamma_1,\theta_1,r_1,k_1](n) +
  s\cdot\mu_2\varphi[\gamma_2,\theta_2,r_2,k_2](t) \label{eq:bound:phi}
  \\
  \text{with }
  \varphi[\gamma,\theta,r,k](t)\defeq \frac{P_k(t)}{k!}
  e^{(t-k)\gamma}\cos((t-k)\theta-r\tfrac{\pi}{2})  \label{eq:bound:varphi}
  \\
  \text{ and } P_k(t)=\prod_{l=0}^{k-1} (t-l) \label{eq:bound:Pk}
\end{gather}
In the sequel, $(\gamma_1,k_1)\succ(\gamma_2,k_2)$
denotes the lexical order
$\gamma_1>\gamma_2\vee\gamma_1=\gamma_2\wedge k_1>k_2$.

\subsection{Technical results}
\label{sec:bound:technical}

The following proposition provides a way to compute the lower and
upper bound of a function on a set of integers, provided some assumptions.

\begin{proposition}[Infinum and supremum of functions on integers]\label{prop:ezero}
  Let $f$ be a continuous function with a continuous first order
  derivative $f'$ and such that:
  \begin{compactenum}[(i)]
  \item $f(\infty)=\lim_{t\rightarrow\infty}f(t)$ exists;
  \item there exists $N\geq 0$
    such that \\
    $(\forall t\sgeq N: f'(t)\sgeq 0) \vee (\forall
    t\sgeq N: f'(t)\sleq 0)$.
  \end{compactenum}
  \smallskip

  Let $N'=\min(N,N_\mathit{max})$. Then
  \begin{align*}
      \inf_{N_\mathit{min}\leq n\leq N_\mathit{max}} f(n) &= \min\big(\min\limits_{N_\mathit{min}\leq n\leq N'}
      f(n),\; f(N_\mathit{max})\big) \\
      \sup_{N_\mathit{min}\leq n\leq N_\mathit{max}} f(n) &= \max\big(\max\limits_{N_\mathit{min}\leq n\leq N'} f(n),\; f(N_\mathit{max})\big)
    \end{align*}
\end{proposition}

For $M>0$ and $(\gamma,k)\succ (0,0)$ ($k$ may be negative), we define
\begin{proposition}[Increasing $f(t)=t^k e^{\gamma
    t}$]\label{prop:poly_frac_greater} ~ \\
  We assume $(\gamma,k)\succ(0,0)$ ($k$ may be negative), $p\geq 2$ and a
  lower bound $M>0$. We define
  \begin{equation}
    \label{eq:bound:T0}
    T_0[M,\gamma,k] \defeq
    \left\{
      \begin{array}{l}
        t_0 + \dfrac{\max(0,M-f(t_0))}{f'(t_0)}
        \quad \text{if } \gamma>0 \\
        \text{with some } t_0\geq
        \begin{cases}
          1 &\ \text{if } k\geq 0 \\
          -\frac{pk}{\gamma} &\text{if } k<0
        \end{cases} \\
        \sqrt[k]{M} \qquad\qquad\qquad \text{if } \gamma=0 \wedge k>0
      \end{array}
    \right.
  \end{equation}
  Then $t\geq T_0[M,\gamma,k] \implies f(t)\geq M$
\end{proposition}
\begin{proof}
  If $\gamma=0$, the result is trivial.
  Otherwise,
  \begin{align*}
    f'(t) &= f(t)\Bigl(\gamma+\frac{k}{t}\Bigr)\\
    f''(t) &= f'(t)\Bigl(\gamma+\frac{k}{t}\Bigr)+f(t)\frac{-k}{t^2} \\
    &= f(t)\left(\Bigl(\gamma+\frac{k}{t}\Bigr)^2-\frac{k}{t^2}\right)
  \end{align*}
  \begin{compactitem}
  \item If $\gamma>0\wedge k\geq 0$, if $t\geq 1$, then $f'(t)$
    and $f''(t)$ are strictly positive.
  \item If $\gamma>0\wedge k<0$, if $t\geq 0$ then
    $f''(t)>0$. Moreover, if $t\geq t_0=-\dfrac{pk}{\gamma}$, then
    $\gamma+\dfrac{k}{t}\geq\dfrac{p-1}{p}\gamma>0$, hence
    $f'(t)>0$.
  \end{compactitem}
  In both cases, from the Taylor
  expansion of order 2, we obtain $f(t) \geq f(t_0)
  + (t-t_0)f'(t_0)$. Hence, $f(t_0) +
  (t-t_0)f'(t_0)\geq M \Leftrightarrow t-t_0\geq
  \frac{M-f(t_0)}{f'(t_0)}$.
\end{proof}

\begin{proposition}[Weighted sums of
  sinus]\label{prop:bound:add_sin} ~ \\
  For any $\mu_1\neq 0, \mu_2\neq 0$, $\varphi_1$, $\varphi_2$,
  \begin{gather}
    \mu_1 \sin(\theta) + \mu_2\cos(\theta)
    = \mu\sin(\theta+\varphi) \\
    \begin{aligned}[t]
      \text{with}\quad \mu &= \sqrt{\mu_1^2+\mu_2^2} \\[-2ex]
      \text{and}\quad \varphi &= \arctan(\mu_2/\mu_1) +
      \begin{cases}
        0 &\text{if } \mu_1\geq 0 \\
        \pi &\text{otherwise}
      \end{cases}
    \end{aligned} \notag \\
    \mu_1 \sin(\theta+\varphi_1) + \mu_2\sin(\theta+\varphi_2)
    = \mu \sin(\theta+\varphi) \\
    \begin{aligned}[t]
      \text{with}\quad \mu &= \sqrt{\mu_1^2+\mu_2^2+2\mu_1\mu_2\cos(\varphi_2-\varphi_1)} \\[-1ex]
      \text{and}\quad\varphi &=
       \begin{multlined}[t]
         \arctan\left(
           \dfrac
           {\mu_1\sin\varphi_1+\mu_2\sin\varphi_2}
           {\mu_1\cos\varphi_1+\mu_2\cos\varphi_2}
         \right) \\
         +
         \begin{cases}
           0 & \text{if } \mu_1\cos\varphi_1+\mu_2\cos\varphi_2 \geq 0 \\
           \pi & \text{otherwise}
         \end{cases}
       \end{multlined}
    \end{aligned}
      \notag
  \end{gather}
\end{proposition}
\begin{proof}
  $\begin{array}[t]{rcl}
    \mu\sin(\theta+\varphi)&=&\mu\sin\theta\cos\varphi + \mu
    \cos\theta\sin\varphi \\
    &=& \mu_1 \sin\theta + \mu_2\cos\theta
  \end{array}$\\
  if $\mu\cos\varphi = \mu_1 \wedge
  \mu\sin\varphi = \mu_2$, which is implied by
  $\mu = \sqrt{\mu_1^2+\mu_2^2}$ and
  $\varphi = \arctan(\mu_2/\mu_1) +
  \begin{cases}
    0 &\text{if } \mu_1\geq 0 \\
    \pi &\text{otherwise}
  \end{cases}
  $.
  \begin{align*}
     &\mu_1 \sin(\theta+\varphi_1) + \mu_2\sin(\theta+\varphi_2) \\
     =& (\mu_1\cos\varphi_1 + \mu_2\cos\varphi_2) \sin\theta +
     (\mu_1\sin\varphi_1 + \mu_2\sin\varphi_2) \cos\theta
  \end{align*}
  and we apply the previous identity.
\end{proof}

\begin{definition}[Decomposition of an angle modulo $\pi$] ~\\
  Given an angle $\varphi$, we define
  \begin{equation*}
    \decomp(\varphi) = (\ell,\phi)
    \quad \text{with} \quad
    \left\{
    \begin{aligned}
      \ell &= \vceil{\frac{\varphi}{\pi}-\frac{1}{2}} \in
      \mathbb{Z} \\
      \phi &= \varphi - \ell\pi \in [-\tfrac{\pi}{2},\tfrac{\pi}{2}]
  \end{aligned}
  \right.
  \end{equation*}
\end{definition}

\begin{proposition}[Extrema of $f(t)=e^{\gamma t}\sin(\theta
  t+\varphi)$]\label{prop:bound:complex:Jp_k_zero} ~\\
  We assume $\theta\in]0,\pi[$ and an interval $[T_\mathit{min},T_\mathit{max}]$.

  \noindent The two first elements $t_0, t_1$ of the set $\{ t
  \;|\; f'(t)=0\wedge t\geq T_\mathit{min}\}$ are defined by
  \begin{gather*}
    \begin{aligned}
      t_i &= \tfrac{1}{\theta}\bigl(-\arctan(\tfrac{\theta}{\gamma}) - \varphi + (\ell_0+i)\pi\bigr) \\
      \text{with }(\ell_\mathit{min},\phi_\mathit{min}) &= \decomp(\theta
      T_\mathit{min}+\varphi) \\
      \text{and } \ell_0 &= \begin{cases}
        \ell_\mathit{min} &\text{if } \phi_\mathit{min}\leq -\arctan(\tfrac{\theta}{\gamma}) \\
        \ell_\mathit{min}+1 & \text{if }  \phi_\mathit{min}>-\arctan(\tfrac{\theta}{\gamma})
      \end{cases}
    \end{aligned}
  \end{gather*}
  The two last elements $t_2, t_3$ of the set $\{ t
  \;|\; f'(t)=0\wedge t\leq T_\mathit{max}\}$ are defined by
  \begin{gather*}
    \begin{aligned}
      t_{2+i} &=
      \tfrac{1}{\theta}\bigl(-\arctan(\tfrac{\theta}{\gamma}) -
      \varphi + (\ell_2-i)\pi\bigr) \\
      \text{with }
      (\ell_\mathit{max},\phi_\mathit{max}) &= \decomp(\theta
      T_\mathit{max}+\varphi) \\
      \text{and }
      \ell_2 &= \begin{cases}
        \ell_\mathit{max}-1 &\text{if } \phi_\mathit{max}\leq -\arctan(\tfrac{\theta}{\gamma}) \\
        \ell_\mathit{max} & \text{if }  \phi_\mathit{max}>-\arctan(\tfrac{\theta}{\gamma})
      \end{cases}
    \end{aligned}
  \end{gather*}
  As a consequence,
  \begin{itemize}
  \item On the interval $[T_\mathit{min},\infty[$, the two
    extrema of $f$ are
    \begin{itemize}
    \item $f(t_0)$ and $f(t_1)$, if $\gamma\leq 0$;
    \item $-\infty$ and $\infty$, if $\gamma>0$.
    \end{itemize}
  \item On the interval $[T_\mathit{min},T_\mathit{max}]$, the two
    extrema of $f$ are
    \begin{multline*}
      \min_{t\in\mathcal{T}} f(t) \text{ and } \max_{t\in\mathcal{T}}
      f(t) \\
      \text{with }
      \mathcal{T}=\{ T_\mathit{min},T_\mathit{max} \} \cup \{
      t_0,t_1,t_2,t_3\} \cap [T_\mathit{min},T_\mathit{max}]
    \end{multline*}
  \end{itemize}
\end{proposition}
\begin{proof}
  $f'(t)=e^{\gamma t}\bigl(\gamma \sin(\theta t + \varphi) + \theta \cos(\theta t + \varphi)\bigr)$.
  $f'(t)=0\Leftrightarrow \tan(\theta t +
  \varphi)=-\tfrac{\theta}{\gamma}$. We decompose $ (\theta
  T_\mathit{min} + \varphi)=\phi_\mathit{min} + \ell_\mathit{min}\pi$ with
  $(\ell_\mathit{min},\phi_\mathit{min})= \decomp(\theta
  T_\mathit{min}+\varphi)$. Looking for the first solution
  $t_0\geq T_\mathit{min}$ of $f(t)=0$, we have
  \begin{gather*}
    \begin{aligned}
      \theta t_0+\varphi &= \theta(t_0-T_\mathit{min}) + \theta T_\mathit{min}+\varphi \\
      &= \theta(t_0-T_\mathit{min}) + \phi + \ell_\mathit{min}\pi
    \end{aligned}\\
    \begin{aligned}
      & \tan(\theta t_0+\varphi) = -\tfrac{\theta}{\gamma} \\
      \Leftrightarrow &
      \tan(\phi + \theta(t_0-T_\mathit{min})) = -\tfrac{\theta}{\gamma} \\
      \Leftrightarrow&
      \begin{cases}
        \phi+\theta(t_0-T_\mathit{min}) = -\arctan(\tfrac{\theta}{\gamma})
        &\text{if } \tan(\phi)\leq -\tfrac{\theta}{\gamma} \\
        \phi+\theta(t_0-T_\mathit{min}) = -\arctan(\tfrac{\theta}{\gamma}) + \pi
        &\text{if } \tan(\phi)> -\tfrac{\theta}{\gamma}
      \end{cases} \\
      \Leftrightarrow&
      \theta t_0+\varphi = -\arctan(\tfrac{\theta}{\gamma}) + \ell_0\pi \\
      & \text{with } \ell_0 =
      \begin{cases}
        \ell_\mathit{min} &\text{if } \phi\leq -\arctan(\tfrac{\theta}{\gamma}) \\
        \ell_\mathit{min}+1 & \text{if }  \phi>-\arctan(\tfrac{\theta}{\gamma})
      \end{cases} \\
      \Leftrightarrow&
      t_0 = \tfrac{1}{\theta}\bigl(-\arctan(\tfrac{\theta}{\gamma}) - \varphi + \ell_0\pi \bigr)
    \end{aligned}
  \end{gather*}
  The second solution is $t_1=t_0+\tfrac{\pi}{\theta}$.
\end{proof}

\begin{proposition}[Extrema of $f(t)=e^{\gamma t}P_k(t)\sin(\theta
  t+\varphi)$]\label{prop:bound:complex:Jp_k_pos} ~\\
  We assume $\gamma\neq 0$, $\theta\in]0,\pi[$, $k>0$ and $P_k(t)$
  defined by Eqn.~(\ref{eq:bound:Pk}). We also assume an interval
  $[T_\mathit{min},T_\mathit{max}]$
  with $T_\mathit{min} \geq N_{\gamma/p,k}$. \\
  The two first elements $t_0, t_1$ of the set $\{ t \;|\;
  f'(t)=0\wedge t\geq T_\mathit{min}\}$ are such that
  \begin{gather*}
    t_0 \in T_0 \cup T'_0 \qquad  t_1\in T_1\cup T_1' \\
      \begin{aligned}[t]
      \text{with}\quad
        T_0 &= \tfrac{1}{\theta}\bigl(\arctan(-R)\cap[\phi_\mathit{min},\tfrac{\pi}{2}]  - \varphi + \ell_\mathit{min}\pi\bigr) \\
        T_0' &= \tfrac{1}{\theta}\bigl(\arctan(-R)\cap[-\tfrac{\pi}{2},\phi_\mathit{min}] -
        \varphi + (\ell_\mathit{min}+1)\pi\bigr)  \\
        T_1 &= T_0+\tfrac{\pi}{\theta} \qquad T'_1 =
        T'_0+\tfrac{\pi}{\theta} \\
        R&=
        \begin{cases}
          \bigl[\tfrac{p}{p+1}\tfrac{\theta}{\gamma}, \tfrac{\theta}{\gamma}
          \bigr] & \text{if } \gamma>0 \\
          \bigl[\tfrac{p}{p-1}\tfrac{\theta}{\gamma}, \tfrac{\theta}{\gamma}
          \bigr] & \text{if } \gamma<0
        \end{cases} \\
        (\ell_\mathit{min}, \phi_\mathit{min}) &=
        \decomp(\theta T_\mathit{min}+\varphi)
      \end{aligned}
   \end{gather*}
   The two last elements $t_2, t_3$ of the set $\{ t
  \;|\; f'(t)=0\wedge t\leq T_\mathit{max}\}$ are such that
  \begin{gather*}
    t_2 \in T_2 \cup T'_2 \qquad  t_3\in T_3 \cup T_3' \\
    \begin{aligned}[t]
      \text{with} \quad
      T_2 &= \tfrac{1}{\theta}\bigl(\arctan(-R)\cap[\phi,\tfrac{\pi}{2}]  - \varphi + (\ell_\mathit{max}-1)\pi\bigr) \\
      T_2' &= \tfrac{1}{\theta}\bigl(\arctan(-R)\cap[-\tfrac{\pi}{2},\phi] -
      \varphi + \ell_\mathit{max}\pi\bigr) \\
      T_3 &= T_2-\tfrac{\pi}{\theta} \qquad T'_3 =
      T'_2-\tfrac{\pi}{\theta} \\
      (\ell_\mathit{max}, \phi_\mathit{max}) &=
      \decomp(\theta T_\mathit{max}+\varphi)
    \end{aligned}
   \end{gather*}
   As a consequence, if\, $T=\bigcup_{0\leq i\leq 3} T_i\cup T'_i$,
   \begin{itemize}
   \item On the interval $[T_\mathit{min},\infty]$, the two extrema
     of $f$ are
     \begin{itemize}
     \item in the set $f(T\cap[T_\mathit{min},\infty])$, if $\gamma<0$;
     \item $-\infty$ and $\infty$ if $\gamma>0$.
     \end{itemize}
   \item On the interval $[T_\mathit{min},T_\mathit{max}]$, the two extrema
     of $f$ are
     in the set $f(T\cap[T_\mathit{min},T_\mathit{max}] \cup \{ T_\mathit{min}, T_\mathit{max} \})$.
   \end{itemize}
   Moreover,
   \begin{itemize}
   \item if $\gamma>0$, $f$ is increasing on the intervals
     composing $T$;
    \item if $\gamma<0$, $f$ is the product of the decreasing term
      $e^{\gamma t}$ and the increasing term $P_k(t)\sin(\theta
      t+\varphi)$ on the same intervals.
  \end{itemize}
  Hence interval arithmetic techniques apply to evaluate the range
  of $f$ on such intervals.
\end{proposition}
\begin{proof}
  $$f'(t) = e^{\gamma t} P_k(t) \Bigl( \Bigl(\gamma+
    \sum_{l=0}^{k-1} \frac{1}{t-l}\Bigr) \sin(\theta t+\varphi) +
    \theta\cos(\theta t+\varphi) \Bigr)$$
  If $t\geq N_{\gamma/p,k}$,
  according to Eqn.~(\ref{eq:N_gamma_k}),
  if $t\geq N_{\gamma/p,k}$, $\sum_{l=0}^{k-1} \frac{1}{t-l} \leq
  \frac{|\gamma|}{p}$, hence
  $$
  \gamma + \sum_{l=0}^{k-1} \frac{1}{t-l} \in R=
  \begin{cases}
    [\gamma,  \tfrac{p+1}{p}\gamma] & \text{if } \gamma>0 \\
    [\gamma, \tfrac{p-1}{p}\gamma] & \text{if } \gamma<0
  \end{cases}
  $$
  \begin{align*}
    f'(t)=0
    & \Leftrightarrow
    \tan(\theta t + \varphi) = \frac{\theta}{\gamma +
      \sum_{l=0}^{k-1} \frac{1}{t-l}} \\
    & \Rightarrow
    \tan(\theta t + \varphi) \in \frac{\theta}{R}
  \end{align*}
  and we reuse the proof of
  Proposition~\ref{prop:bound:complex:Jp_k_zero}, exploiting the fact
  that the function $\tan(\theta t + \varphi)$ is increasing if $\theta t+\varphi \in [-\tfrac{\pi}{2},\tfrac{\pi}{2}]+\ell_\mathit{min}\pi$.
\end{proof}
\begin{proposition}[Extrema of $f(t)=P_k(t)\sin(\theta
  t+\varphi)$]\label{prop:bound:complex:Jp_k_pos2} ~ \\
  We assume
  $\theta\in]0,\pi[$, $k>0$ and $P_k(t)$ defined by
  Eqn.~(\ref{eq:bound:Pk}).
  We also assume an interval $[T_\mathit{min},T_\mathit{max}]$ with $T_\mathit{min} \geq k$. \\
  The two first elements $t_0, t_1$ of the set $\{ t
  \;|\; f'(t)=0\wedge t\geq T_\mathit{min}\}$ are such that
  \begin{gather*}
    t_0 \in T_0 \cup T'_0 \qquad  t_1\in T_1\cup T'_1 \\
    \begin{aligned}
      \text{with } &
      R=  \tfrac{\theta}{k}
      [T_\mathit{min}-k+1,
      \min(
      T_\mathit{max},
      T_\mathit{min}+\tfrac{2\pi}{\theta}
      )
      ] \\
      \text{and } & T_0,T'_0,T_1,T'_1 \text{ defined as in
        Proposition~\ref{prop:bound:complex:Jp_k_zero}}
    \end{aligned}
   \end{gather*}
   Symmetrically, the two last elements $t_2, t_3$ of the set $\{ t
  \;|\; f'(t)=0\wedge t\leq T_\mathit{max}\}$ are such that
  \begin{gather*}
    t_2 \in T_2 \cup T'_2 \qquad  t_3\in T_3\cup T'_3 \\
    \begin{aligned}
      \text{with } &
      R= \tfrac{\theta}{k}
      [
      \max(T_\mathit{min},
      T_\mathit{max}-\tfrac{2\pi}{\theta})-k+1,
      T_\mathit{max}
      ] \\
      \text{and } & T_2,T'_2,T_3,T'_3 \text{ defined as in
        Proposition\ref{prop:bound:complex:Jp_k_zero}}
    \end{aligned}
  \end{gather*}
\end{proposition}
\begin{proof}
  $$f'(t) = \alpha P_k(t) \Bigl( \Bigl(
    \sum_{l=0}^{k-1} \frac{1}{t-l}\Bigr) \sin(\theta t+\varphi) +
    \theta\cos(\theta t+\varphi) \Bigr)$$
    The two first solutions of $f'(t)=0$ in
    $[T_\mathit{min},T_\mathit{max}]$ are in
    $[T_\mathit{min},T_\mathit{min}
    +\tfrac{2\pi}{\theta}]\cap[T_\mathit{min},T_\mathit{max}]=
    [T_\mathit{min},T_\mathit{min}']$ with
    $T_\mathit{min}'=\min\bigl( T_\mathit{min}
    +\tfrac{2\pi}{\theta}, T_\mathit{max} \bigr)$. %
    On this interval, $\sum_{l=0}^{k-1} \frac{1}{t-l}\in \left[
      \frac{k}{T_\mathit{min}'},
      \frac{k}{T_\mathit{min}-k+1}\right]$ and
    $
    \frac{\theta}{\sum_{l=0}^{k-1} \frac{1}{t-l}}\in R\defeq
    \left[
      \tfrac{T_\mathit{min}-k+1}{k}\theta,
      \tfrac{T_\mathit{min}'}{k}\theta
    \right]
    $.
    We then proceed similarly as in
    proposition~\ref{prop:bound:complex:Jp_k_pos}.
    Symmetrically, the two last solutions of
    $f'(t)=0$ in $[T_\mathit{min},T_\mathit{max}]$ are in
    $[T_\mathit{max}-\tfrac{2\pi}{\theta},T_\mathit{max}]
    \cap[T_\mathit{min},T_\mathit{max}]=
    [T_\mathit{max}',T_\mathit{max}]$ with
    $T_\mathit{max}'=\max\bigl( T_\mathit{min},
     T_\mathit{max}-\tfrac{2\pi}{\theta}\bigr)$.

    On this interval, $\sum_{l=0}^{k-1} \frac{1}{t-l}\in \left[
      \frac{k}{T_\mathit{max}},
      \frac{k}{T_\mathit{max}'-k+1}\right]$ and
    $
    \frac{\theta}{\sum_{l=0}^{k-1} \frac{1}{t-l}}\in R\defeq
    \left[
      \tfrac{T_\mathit{max}'-k+1}{k}\theta,
      \tfrac{T_\mathit{max}}{k}\theta
    \right]
    $.
\end{proof}

\subsection{Expressions $\phi(n)$ involving only real and positive
  eigenvalues}
\label{sec:bound:realpositive}

Coming back to Eqn.~(\ref{eq:bound:phin}), we consider here the subcase $\lambda_i>0\wedge \theta_i=0$ for
$i=1,2$, which leads to
the simpler expressions
\begin{gather*}
  \phi(t)=\mu_1\varphi_{\gamma_1,k_1}(t) +
  s\mu_2\varphi_{\gamma_2,k_2}(t) \\
  \text{ with }
  \varphi_{\gamma,k}(n)\defeq \frac{P_k(t)}{k!}
  e^{(t-k)\gamma}
\end{gather*}
We will apply Proposition~\ref{prop:ezero} to the function $\phi$, which is the
linear combination of two smooth functions
$\varphi_{\gamma_i,k_i}$ with the following properties. By
defining
\begin{equation}
  \label{eq:N_gamma_k}
  N_{\gamma,k}\defeq k+\bigl\lceil\tfrac{k}{|\gamma|}\bigr\rceil
  \quad\text{for } \gamma\neq 0
\end{equation}
we have
\begin{gather}
  \lim_{t\rightarrow\infty} \varphi_{\gamma,k} =
  \begin{cases}
    0 & \text{ if } \gamma<0 \\
    1 & \text{ if } (\gamma,k)=(0,0) \\
    \infty & \text{ if } (\gamma,k)\succ (0,0)
  \end{cases}  \displaybreak[1] \\
  \varphi'_{\gamma,k}(t) = \varphi_{\gamma,k}(t)\Bigl(\gamma + \dfrac{P'_k(t)}{P_k(t)}\Bigr) \qquad \dfrac{P'_k(t)}{P_k(t)} = \sum_{l=0}^{k-1}\frac{1}{t-l}
  \label{eq:varphip_gamma_k}
  \displaybreak[1] \\
  t\geq k \implies \frac{k}{t}\leq \sum_{l=0}^{k-1} \frac{1}{t-l} \leq \frac{k}{t-k+1} \label{eq:bound:logPp} \displaybreak[1] \\
  \gamma\neq 0 \wedge t\geq N_{\gamma,k} \implies \sum_{l=0}^{k-1} \frac{1}{t-l} < |\gamma|
  \label{eq:N_gamma_k_2}\displaybreak[1] \\
  \begin{multlined}[b]
    \gamma\geq 0\wedge t\geq N\geq k \implies \\
    \gamma + \tfrac{k}{t}\leq \Bigl| \gamma + \sum_{l=0}^{k-1} \frac{1}{t-l} \Bigr| \leq
    \gamma + \tfrac{k}{t-k+1} \leq
    \gamma+\tfrac{k}{N-k+1}
  \end{multlined}
  \label{eq:bound:gamma_p}\displaybreak[1] \\
  \begin{multlined}[b]
    \shoveright[3cm]{\gamma<0 \wedge t\geq N\geq N_{\gamma,k} \implies } \\
    \shoveright[1cm]{ 0<|\gamma|-\tfrac{k}{N-k+1}\leq |\gamma|-\tfrac{k}{t-k+1} \leq} \\
    \Bigl| \gamma + \sum_{l=0}^{k-1} \frac{1}{t-l} \Bigr| \leq
    |\gamma|-\tfrac{k}{t}\leq |\gamma|
  \end{multlined}
  \label{eq:bound:gamma_n}\displaybreak[1]  \\
  \begin{array}{lll}
    \gamma<0\wedge t\geq N_{\gamma,k} &\implies& \varphi'_{\gamma,k}(t)<0 \\
    \gamma\geq 0 &\implies& \varphi'_{\gamma,k}(t)\geq 0 \\
    (\gamma,k)\succ(0,0) &\implies& \varphi'_{\gamma,k}(t)> 0
  \end{array}
  \label{eq:bound:varphip_lim}\displaybreak[1]
\end{gather}
Our aim is thus, in the sequel of this subsection, to effectively
compute $N$ such for any $t\geq N$, $\phi'(t)$ has a constant
sign, so as to apply Prop.~\ref{prop:ezero}.

We eliminate the following cases:
\begin{itemize}
\item Case $\mu_2=0$: in this case $\phi(t)=\mu_1\varphi_{\gamma,k}(t)$ and we
  take $N=N_{\gamma,k}$ according to Eqns.~(\ref{eq:N_gamma_k_2})
  and (\ref{eq:bound:varphip_lim}).
\item Case $(\gamma_1,k_1)=(0,0) \vee (\gamma_2,k_2)=(0,0)$: as
  $\varphi_{0,0}$ is constant, we fall back to the case $\mu_2=0$.
\end{itemize}
We also assume
$$
t\geq N\defeq \max(N_{\gamma_1,k_1},N_{\gamma_2,k_2},k_1+1,2k_2+1)
$$
where $N_{\gamma,k}$ is defined by
Eqn.~(\ref{eq:N_gamma_k}). (Reasons for this assumption will appear later).

We consider the derivative
$\phi'(t)=\mu_1\varphi'_{\gamma_1,k_1}(t) +
s\mu_2\varphi'_{\gamma_2,k_2}(t)$ and the ratio of the absolute
value of its two terms of the derivative:
\begin{multline}
  \label{eq:bound:ratio}
  r(t)=r[\mu_1,\gamma_1,k_1,\mu_2,\gamma_2,k_2](t) \\
  =
  \frac{\mu_1}{\mu_2}\frac{k_2!}{k_1!}\frac{P_{k_1}(t)}{P_{k_2}(t)}
  \dfrac{
    \big|\gamma_1+\sum_{l=0}^{k_1-1} \frac{1}{t-l}\big|
  }{
    \big|\gamma_2+\sum_{l=0}^{k_2-1} \frac{1}{t-l}\big|
  }
  e^{(\gamma_1-\gamma_2)t+\gamma_2k_2-\gamma_1k_1} \\
  =
  R
  \dfrac{\prod_{l=0}^{k_1-1} (t-l)}{\prod_{l=0}^{k_2-1} (t-l)}
  \dfrac{
    \big|\gamma_1+\sum_{l=0}^{k_1-1} \frac{1}{t-l}\big|
  }{
    \big|\gamma_2+\sum_{l=0}^{k_2-1} \frac{1}{t-l}\big|
  }
  e^{(\gamma_1-\gamma_2)t}
\end{multline}
with the constant
\begin{gather}\label{eq:bound:R}
  R=R[\mu_1,\gamma_1,k_1,\mu_2,\gamma_2,k_2] =
  \frac{\mu_1}{\mu_2}\frac{k_2!}{k_1!}e^{\gamma_2k_2-\gamma_1k_1}>0
\end{gather}
Our hypothesis together with Eqn.~(\ref{eq:bound:varphip_lim}) ensure that both
numerator and denominator of $r(t)$ are strictly positive.
This ratio indicates which of the two terms defining the derivative
is greater or equal to the other in absolute value, and allows us to
deduce the sign of the derivative.
\begin{proposition}[Symmetries for $\varphi(t)$ and $r(t)$] \label{prop:varphi:symmetry}
  \begin{align*}
    \phi[\mu_1,\gamma_1,k_1,s,\mu_2,\gamma_2,k_2](t)&= s\cdot
  \phi[\mu_2,\gamma_2,k_2,s,\mu_1,\gamma_1,k_1](t) \\
    r[\mu_1,\gamma_1,k_1,\mu_2,\gamma_2,k_2](t) &=
    1/r[\mu_2,\gamma_2,k_2,\mu_1,\gamma_1,k_1](t)
  \end{align*}
\end{proposition}
We can thus focus on the case
$(\gamma_1,k_1)\succ(\gamma_2,k_2)$.

The following proposition aims at finding lower bounds for the
multiplicative terms of $r(t)$.
\begin{proposition}\label{prop:ratio:polynom}
  With our hypothesis on $t$,
  \begin{gather}
    \dfrac{\prod_{l=0}^{k_1-1} (t-l)}{\prod_{l=0}^{k_2-1} (t-l)} \geq
    \dfrac{t^{k_1-k_2}}{2^{k_1}}
  \end{gather}
\end{proposition}
\begin{proof}
  The hypothesis $t\geq 2k_1+1$ implies $t-k_1\geq t/2$.
  If $k_1\geq k_2$, $\dfrac{\prod_{l=0}^{k_1-1}
    (t-l)}{\prod_{l=0}^{k_2-1} (t-l)} = {\displaystyle \prod_{l=k_2}^{k_1-1}
  (t-l)} \geq (t-k_1)^{k_1-k_2} \geq (t/2)^{k_1-k_2} \geq \dfrac{t^{k_1-k_2}}{2^{k_1}}$.
  If $k_1<k_2$, $\dfrac{\prod_{l=0}^{k_1-1}
    (t-l)}{\prod_{l=0}^{k_2-1} (t-l)} = {\displaystyle
    \prod_{l=k_1}^{k_2-1} \frac{1}{t-l}} \geq
  \bigl(\dfrac{1}{t}\bigr)^{k_2-k_1}=
  t^{k_1-k_2}\geq
  \dfrac{t^{k_1-k_2}}{2^{k_1}}$.
\end{proof}
We define for $\gamma_1\geq\gamma_2 \wedge \neg(\gamma_1=0>\gamma_2)$
\begin{multline}
  \label{eq:ratio:L}
  L[\gamma_1,k_1,\gamma_2,k_2]\defeq \\
  \begin{cases}
    \dfrac{\gamma_1}{\gamma_2+\frac{k_2}{N-k_2+1}}
    &\text{if } \gamma_1\geq \gamma_2\geq 0 \wedge \gamma_1>0\\[0ex]
    k_1/2k_2
    &\text{if } \gamma_1=\gamma_2=0 \\[0ex]
    \dfrac{|\gamma_1|-\frac{k_1}{N-k_1+1}}{|\gamma_2|}
    &\text{if } 0>\gamma_1\geq\gamma_2 \\[0ex]
    \gamma_1/|\gamma_2|
    &\text{if } \gamma_1> 0>\gamma_2
  \end{cases}
\end{multline}

\begin{proposition}\label{prop:ratio:dpolynom}
  With our hypothesis on $t$,
  \begin{itemize}
  \item if $\gamma_1\geq\gamma_2 \wedge \neg(\gamma_1=0>\gamma_2)$,
    \begin{equation}\label{eq:ratio:dpolynom}
      \frac
      {\bigl| \gamma_1 + \sum_{l=0}^{k_1-1} \frac{1}{t-l} \bigr|}
      {\bigl| \gamma_2 + \sum_{l=0}^{k_2-1} \frac{1}{t-l} \bigr|}
      \geq
      L[\gamma_1,k_1,\gamma_2,k_2]
    \end{equation}
  \item if $\gamma_1=0>\gamma_2$,
    $\frac
    {\bigl| \gamma_1 + \sum_{l=0}^{k_1-1} \tfrac{1}{t-l} \bigr|}
    {\bigl| \gamma_2 + \sum_{l=0}^{k_2-1} \tfrac{1}{t-l} \bigr|}
    \geq \dfrac{k_1}{|\gamma_2|t}$
  \end{itemize}
\end{proposition}
\begin{proof}
  The second case of Eqns.~(\ref{eq:ratio:L}--(\ref{eq:ratio:dpolynom}) comes from
  $\frac {\bigl| \gamma_1 + \sum_{l=0}^{k_1-1} \tfrac{1}{t-l}
    \bigr|} {\bigl| \gamma_2 + \sum_{l=0}^{k_2-1} \tfrac{1}{t-l}
    \bigr|} \geq \frac{k_1}{k_2}\dfrac{t-k_2+1}{t} \geq
  \frac{k_1}{2k_2}$ for $t\geq 2k_2+1$.

  The other
  cases of Eqn.~(\ref{eq:ratio:dpolynom}) follow from
  Eqns.~(\ref{eq:bound:gamma_p})--(\ref{eq:bound:gamma_n}).
\end{proof}
By combining
Propositions~\ref{prop:ratio:polynom}-\ref{prop:ratio:dpolynom}, we
obtain the following proposition.
\begin{proposition}\label{prop:ratio:lower_bound}
  We assume $(\gamma_1,k_1)\succ(\gamma_2,k_2)$. With our
  hypothesis on $t$,
  \begin{itemize}
  \item If $\neg (\gamma_1=0>\gamma_2)$,
      $$r(t)\geq
      1/2^{k_1} R
      Lt^{k_1-k_2}e^{(\gamma_1-\gamma_2)t}$$
  with $R$ and $L$ are defined by Eqns.~(\ref{eq:bound:R})
  and (\ref{eq:ratio:L});
\item If $\gamma_1=0>\gamma_2$,
  $$r(t)\geq 1/2^{k_1}
  R
  \frac{k_1}{\,|\gamma_2|\,}
  t^{k_1-k_2-1}e^{(\gamma_1-\gamma_2)t}$$
\end{itemize}
\end{proposition}

We define for $M>0$ and $(\gamma_1,k_1)\succ(\gamma_2,k_2)$
\begin{multline}
    N_0=N_0[M,\mu,\gamma_1,k_1,\gamma_2,k_2]\defeq \\
    \begin{cases}
      \left\lceil
        T_0\left[
          \dfrac{1}{L}\dfrac{2^{k_1}M}{R},\;
          \gamma_1-\gamma_2,\;
          k_1-k_2
        \right]
      \right\rceil \\
      \quad\text{if } \neg(\gamma_1=0>\gamma_2) \\[2ex]
       \left\lceil
        T_0\left[
          \dfrac{|\gamma_2|}{k_1}\dfrac{2^{k_1}M}{R},\;
          -\gamma_2,\;
          k_1-k_2-1
        \right]
      \right\rceil\\
      \quad\text{if } \gamma_1=0>\gamma_2
    \end{cases}
  \end{multline}
  with $T_0$, $R$ and $L$ defined by Eqns.~(\ref{eq:bound:T0}),
  (\ref{eq:bound:R}) and (\ref{eq:ratio:L}).
\begin{proposition}[Properties of $r(t)$] \label{prop:r:prop} ~\\
  $\forall t\geq N_0[M,\mu,\gamma_1,k_1,\gamma_2,k_2]: r(t)\geq M$.
\end{proposition}
\begin{proof}
  We combine Proposition~\ref{prop:poly_frac_greater} and Proposition~\ref{prop:ratio:lower_bound}.
\end{proof}
\begin{theorem}[Properties of $\phi'(t)$] \label{thm:phip}
  By
  applying Proposition~\ref{prop:r:prop} with $M>1$, given
  $N=N[M,\mu_1,\gamma_1,k_1,\mu_2,\gamma_2,k_2]$ as defined there, for all $t\geq N$,
  $$
  \begin{array}{rcl}
    (\gamma_1,k_1)\succ(\gamma_2,k_2)
    &\implies&
    \left\{
      \begin{array}{rcl}
        \gamma_1\geq 0 &\implies \phi'(t)>0 \\
        \gamma_1< 0 &\implies \phi'(t)<0
      \end{array}
    \right. \\
    (\gamma_1,k_1)\prec(\gamma_2,k_2)
    &\implies&
    \left\{
      \begin{array}{rcl}
        \gamma_2\geq 0 &\implies& s\cdot \phi'(t)>0 \\
        \gamma_2<0 &\implies& s\cdot \phi'(t)<0
      \end{array}
    \right.
  \end{array}
$$
This allows us to apply Proposition~\ref{prop:ezero} to compute the
extrema of $\phi(n)$ on any interval of positive integers.
\end{theorem}
\begin{proof}
  The function $r(t)$ indicates which of the two terms of the
  function $\phi'(t)$ dominate in absolute value. $\forall t\geq N:
  r(t)>1$ indicates that starting from the rank $N$, the first
  term dominates. We then have to look at the sign of this dominant
  term. The same holds for the case $\forall t\geq N: r(t)<1$.
\end{proof}

\subsection{Expressions $\phi(n)$ involving only real, non-zero eigenvalues}
\label{sec:bound:realnegative}

Coming back to Eqn~(\ref{eq:bound:phi}),
we consider now the subcase $\lambda_i>0, \theta_i\in\{0,\pi\}$ for
$i=1,2$, with at least one the $\theta_i$'s equal to $\pi$.

Let $s_i=\cos \theta_i\in\{-1,1\}$.
We define
\begin{align*}
  \phi_0(t) &= s_1^{k_1} \phi[\mu_1,\gamma_1,s s_1^{k_1}s_2^{k_2},\mu_2,\gamma_2](2t) \\
  \phi_1(t) &= s_1^{k_1+1} \phi[\mu_1,\gamma_1,s s_1^{k_1+1}s_2^{k_2+1},\mu_2,\gamma_2](2t)
\end{align*}
with $\phi[\cdot]$ defined by Eqn.~(\ref{eq:bound:phi}). The idea is to
separate the cases $n-k_i$ is odd and $n-k_i$ is even. We have
\begin{gather*}
  \phi(n) =
  \begin{cases}
    \phi_0(n/2) & \text{if } n\equiv 0\pmod 2 \\
    \phi_1(n/2) & \text{if } n\equiv 1\pmod 2
  \end{cases}
  \intertext{and}
  \begin{multlined}
    \inf_{N_\mathit{min}\leq n\leq N_\mathit{max}} \phi(n) = \\
    \min\left(\inf_{\vceil{\frac{N_\mathit{min}}{2}}\leq n\leq \vfloor{\frac{N_\mathit{max}}{2}}} \phi_0(n),
    \inf_{\vceil{\frac{N_\mathit{min}+1}{2}}\leq n\leq \vfloor{\frac{N_\mathit{max}+1}{2}}} \phi_1(n)\right)
  \end{multlined}\\
  \begin{multlined}
    \sup_{N_\mathit{min}\leq n\leq N_\mathit{max}} \phi(n) = \\
    \max\left(\sup_{\vceil{\frac{N_\mathit{min}}{2}}\leq n\leq \vfloor{\frac{N_\mathit{max}}{2}}} \phi_0(n),
    \sup_{\vceil{\frac{N_\mathit{min}+1}{2}}\leq n\leq \vfloor{\frac{N_\mathit{max}+1}{2}}} \phi_1(n)\right)
  \end{multlined}
\end{gather*}
We compute the integer extrema of $\phi_1$ and $\phi_2$ by
applying the technique of Section~\ref{sec:bound:realpositive}.

\subsection{Expressions $\phi(n)$ involving only complex, non-zero eigenvalues}
\label{sec:bound:complex}

Coming back to Eqn~(\ref{eq:bound:phi}), we consider now the
subcase $\lambda_i>0, \theta_i\in]0,\pi[$ for $i=1,2$.

In contrast to Sections~\ref{sec:bound:realpositive} and
\ref{sec:bound:realnegative}, we will mostly compute 
infimum and supremum on reals instead of on integers (which is an
approximation of the initial problem).  Moreover, 
we will give approximate bounds for the general case and more
precise ones only in special cases.

A trivial but sound bound for $N_\mathit{min}\leq n\leq N_\mathit{max}$ is
\begin{align*}
  \label{eq:bound:complex:simple}
  & -M\leq \phi(n)\leq M \\
  & \text{ with }
  \begin{multlined}[t]
    M=\sup_{N_\mathit{min}\leq n\leq N_\mathit{max}}
    \phi[\mu_1,\gamma_1,\theta_1=0,r_1=0,k_1,\\[-2ex]
    s=1,\mu_2,\gamma_2,\theta_2=0,k_2](n)
  \end{multlined}
\end{align*}
where $M$ is the upper bound computed in
Section~\ref{sec:bound:realpositive}.

In the sequel, we improve such bounds in some cases.

\subsubsection{Bounding single coefficients}
\label{sec:bound:complex:single}

We consider here $\mu=1$, which means that $\phi(t)=\varphi(t)$ as
defined by Eqn.~(\ref{eq:bound:varphi}). We have
\begin{align}
  &\varphi[\gamma,\theta,r,k](t)=\frac{P_k(t)}{k!}
  e^{(t-k)\gamma}\cos((t-k)\theta-r\tfrac{\pi}{2})
\end{align}
For bounding $\varphi(t)$ on some interval $[N_\mathit{min},N_\mathit{max}]$,
we define $\alpha=\tfrac{e^{-\gamma k}}{k!}$,
$\varphi=-\theta k+(1-r)\tfrac{\pi}{2}$. We have
$$
\varphi[\gamma,\theta,r,k](t)=\alpha P_k(t) e^{\gamma t}\sin(\theta t+\varphi)
$$
\begin{itemize}
\item If $k=0$, we apply Proposition~\ref{prop:bound:complex:Jp_k_zero};
\item If $k>0 \wedge \gamma\neq 0$, we apply Proposition~\ref{prop:bound:complex:Jp_k_pos};
\item If $k>0 \wedge \gamma=0$, we apply
  Proposition~\ref{prop:bound:complex:Jp_k_pos2}.
\end{itemize}
Each of these propositions assume a minimum lower (integer) bound $L$ on
$T_\mathit{min}$: in practice, we enumerate the values of
$\varphi(t)$ on the integers in $[N_\mathit{min},
\min(L-1,N_\mathit{max})]$, and we resort to the relevant proposition
for the interval $[L, N_\mathit{max}]$.

\subsubsection{Bounding some linear combination of two coefficients}
\label{sec:bound:complex:two}

There are some cases where we can do better than
Eqn.~(\ref{eq:bound:complex:simple}).

\begin{proposition}
  If $\theta_1=\theta_2$, $r_1=r_2$ and $k_1=k_2$, we have
  $$
  \phi(t)=\phi[\mu_1,\gamma_1,k_1,s,\mu_2,\gamma_2,k_2]
  \cos(\theta_1(t-k_1) + r_1\frac{\pi}{2})
  $$
  and we can take $\phi(n)\in I_1\cdot I_2$ with
  \begin{align*}
    I_1 &= [\inf V_1,\sup V_1] & \text{with }
    V_1 =\;& \{
    \phi[\mu_1,\gamma_1,k_1,s,\mu_2,\gamma_2,k_2](n) \;|\;\\
    &&&
    \quad n\in[N_\mathit{min},N_\mathit{max}] \} \\
    I_2 &= [\inf V_2,\sup V_2] & \text{with }
    V_1 =\;& \{ cos(\theta_1(t-k_1) + r_1\frac{\pi}{2}) \;|\; \\
    &&&
    \quad t \in [N_\mathit{min},N_\mathit{max}] \}
  \end{align*}
  which improves
  Eqn.~(\ref{eq:bound:complex:simple}) when $s=-1$ or $I_2\subsetneq[-1,1]$.
\end{proposition}
The above case is less likely to occur in practic than the following
one that corresponds to the linear combination of the real and
imaginary part of the same complex eigenvalue.
\begin{proposition}\label{prop:bound:complex:two:associated}
  We assume $\gamma_1=\gamma_2=\gamma$, $k_1=k_2=k$, and
  $\theta_1=\theta_2=\theta$ (which implies $r_1\neq r_2$). We
  have
  \begin{align*}
    &\phi(t)=\mu \frac{P_k(t)}{k!}e^{\gamma(t-k)} \sin(\theta t +
    \varphi) \\
    & \text{ with }
    \begin{aligned}[t]
      \mu &= \sqrt{\mu_1^2+\mu_2^2} \\
      \varphi &=
      \begin{cases}
        \begin{multlined}[t]
          \arctan\left(\dfrac{-\mu_1\sin(\theta k)+s\mu_2\cos(\theta k)}{\mu_1\cos(\theta k)+s\mu_2\sin(\theta k)}\right) \\[-2ex]
          +
          \begin{cases}
            0 &\text{if } \mu_1\cos(\theta k)+s\mu_2\sin(\theta k)\geq 0 \\
            \pi &\text{otherwise}
          \end{cases}
        \end{multlined}
        & \text{if } r_1=0 \\[2ex]
        \begin{multlined}[t]
          \arctan\left(\dfrac{\mu_1\cos(\theta k)-s\mu_2\sin(\theta k)}{\mu_1\sin(\theta k)+s\mu_2\cos(\theta k)}\right) \\[-2ex]
          +
          \begin{cases}
            0 &\text{if } \mu_1\sin(\theta k)+s\mu_2\cos(\theta k)\geq 0 \\
            \pi &\text{otherwise}
          \end{cases}
        \end{multlined}
        & \text{if } r_1=1 \\
      \end{cases}
    \end{aligned}
  \end{align*}
  We can then resort to Section~\ref{sec:bound:complex:single}
\end{proposition}
\begin{proof}
\begin{align*}
  (k!) \phi(t)&=
  P_{k}(t)e^{\gamma(t-k)}
  \begin{multlined}[t]
    (\mu_1\sin(\theta t - \theta k + r_1\tfrac{\pi}{2})+\\
    s\mu_2\sin(\theta t - \theta k + r_2\tfrac{\pi}{2}))
  \end{multlined}
  \\
  &= P_{k}(t)e^{\gamma(t-k)} \mu \sin(\theta t + \varphi)
  \intertext{with}
  \mu &= \sqrt{\mu_1^2+\mu_2^2+2s\mu_1\mu_2\cos((r_1-r_2)\tfrac{\pi}{2})} \\
  &= \sqrt{\mu_1^2+\mu_2^2} \\
  \varphi &=
  \arctan\left(\dfrac{\mu_1\sin(-\theta k\sadd r_1\frac{\pi}{2})\sadd s\mu_2\sin(-\theta k\sadd r_2\frac{\pi}{2})}{\mu_1\cos(-\theta k\sadd r_1\frac{\pi}{2})\sadd s\mu_2\cos(-\theta k\sadd r_2\frac{\pi}{2})}\right) \\
  & +
  \begin{cases}
    0 &\text{if } \mu_1\cos(-\theta k\sadd r_1\frac{\pi}{2})\sadd s\mu_2\cos(-\theta k\sadd r_2\frac{\pi}{2}) \geq 0 \\
    \pi &\text{otherwise}
  \end{cases}
\end{align*}
If $r_1=0$ (which implies $r_2=1$),
\begin{align*}
  & \dfrac{\mu_1\sin(-\theta k\sadd r_1\frac{\pi}{2})\sadd s\mu_2\sin(-\theta k\sadd r_2\frac{\pi}{2})}{\mu_1\cos(-\theta k\sadd r_1\frac{\pi}{2})\sadd s\mu_2\cos(-\theta k\sadd r_2\frac{\pi}{2})} \\
  = &
  \dfrac{\mu_1\sin(-\theta k)+s\mu_2\cos(-\theta k)}{\mu_1\cos(-\theta k)-s\mu_2\sin(-\theta k)} \\
  = &
  \dfrac{-\mu_1\sin(\theta k)+s\mu_2\cos(\theta k)}{\mu_1\cos(\theta k)+s\mu_2\sin(\theta k)}
\end{align*}
If $r_1=1$ (which implies $r_2=0$),
\begin{align*}
  & \dfrac{\mu_1\sin(-\theta k\sadd r_1\frac{\pi}{2})\sadd s\mu_2\sin(-\theta k\sadd r_2\frac{\pi}{2})}{\mu_1\cos(-\theta k\sadd r_1\frac{\pi}{2})\sadd s\mu_2\cos(-\theta k\sadd r_2\frac{\pi}{2})} \\
  = &
  \dfrac{\mu_1\cos(-\theta k)+s\mu_2\sin(-\theta k)}{-\mu_1\sin(-\theta k)+s\mu_2\cos(-\theta k)} \\
  = &
  \dfrac{\mu_1\cos(\theta k)-s\mu_2\sin(-\theta k)}{\mu_1\sin(\theta k)+s\mu_2\cos(\theta k)}
\end{align*}
\end{proof}

\subsubsection{Other cases}
\label{sec:bound:other}

As if $\lambda=0$, $\varphi_{\lambda,k}(n)$ is
constant starting from rank $n>k$, if $\lambda_1=\lambda_2=0$,
computing the extremum is trivial. If only one of the
$\lambda_i$s is zero, we can reduce the problem to one of the previous cases.

\section{Computing a maximum number of iterations}
\label{sec:maxiter}

The goal of this section is to compute
a variant of Eqn.~(\ref{eq:minNumberIt}), rewritten as
\begin{equation}\label{eq:maxiter:phi}
\min \{ n\geq N_0 \;|\; \phi(n) > e_0 \}
\end{equation}
with $\phi(n)$ defined as in Eqn.~(\ref{eq:bound:phin}).  We
generalize Eqn.~(\ref{eq:minNumberIt}) by computing the minimum
starting from some rank $N_0$ instead of $0$, but we restrict the
type of functions to $\phi(n)$.

\subsection{Expressions $\phi(n)$ involving only real, positive eigenvalues}
\label{sec:maxiter:realpositive}

We know from Section~\ref{sec:bound:realpositive} that starting
from some rank $N$, $\phi(t)$ is monotone.

We first compute by enumeration the finite set $S=\{ n \;|\; N_0\leq n\leq N \wedge
  \phi(n)> e_0 \}$.
\begin{itemize}
\item If $S\neq\emptyset$, we select its minimum
  and we have the solution to Eqn.~(\ref{eq:maxiter:phi}).
\item Otherwise we define $N'=\max(N_0,N)$. We know that $\phi(N')\leq e_0$.
  \begin{itemize}
  \item If $\phi(t)$ is increasing for $t\geq N'$ and
    $\phi(\infty)>e_0$, then there exists a unique solution $t_0$
    to the equation $\phi(t)=e_0$, that we can compute by
    the Newton-Raphson method, and $n_0 =
    \begin{cases}
      t_0+1 &\text{ if } t_0\in\mathbb{N} \\
      \ceil{t_0} &\text{ otherwise}
    \end{cases}
    $ is the solution to our problem.
  \item Otherwise, if either $\phi(t)$ is strictly decreasing for $t\leq N'$
    or $\phi(\infty)\leq e_0$, then the solution is $\min\emptyset=\infty$.
  \end{itemize}
\end{itemize}

\subsection{Expressions $\phi(n)$ involving only complex, non-zero eigenvalues}
\label{sec:maxiter:complex}

\begin{figure*}[t]
  \centering

  \parbox[t]{0.3\linewidth}{
    \vspace*{1ex}

    \psset{xunit=0.33em,yunit=4em,plotpoints=50,Dx=100,Dy=100}
    \begin{pspicture}(-10,-1)(40,1.3)
      \psaxes(0,0)(-10,0)(40,0)
      \psplot{-10}{40}{x dup 11 sub 14 mul sin exch 0.985 exch exp mul
        0.95 mul}
      \psline[linestyle=dashed](-10,-0.3)(40,-0.3)
      \rput[b]{0}(-7,-0.28){$e_0$}
      \psline[linestyle=dotted](0,0.7)(0,0)
      \rput[b]{0}(0,0.7){$N_0$}
      \psline[linestyle=dotted](11,1)(11,0)
      \rput[l]{0}(12,0.7){$t_0$}
      \psline[linestyle=dotted](9,1)(9,-0.3)
      \rput[r]{0}(8,0.7){$t'$}
      \psline[linestyle=dotted](13,0)(12,-0.7)
      \rput[t]{0}(13,-0.7){$n_0$}
      \psline[linestyle=dotted](24,1)(24,0)
      \psline[arrows=<->,linestyle=solid](11,1)(24,1)
      \rput[b]{0}(17,1.02){$\frac{\pi}{\theta}>1$}
    \end{pspicture}
    \caption{Case $e_0<0$}
    \label{fig:maxiter:solve:case1}
  }\hfill
  \parbox[t]{0.3\linewidth}{
    \vspace*{1ex}

    \psset{xunit=0.33em,yunit=4em,plotpoints=50,Dx=100,Dy=100}
    \begin{pspicture}(-10,-1)(40,1.3)
      \psaxes(0,0)(-10,0)(40,0)
      \psplot{-10}{40}{x dup 5 sub 14 mul sin exch 0.985 exch exp mul
        0.95 mul}
      \psline[linestyle=dashed](-10,0.5)(40,0.5)
      \rput[b]{0}(-7,0.52){$e_0$}
      \psline[linestyle=dotted](0,0.7)(0,0)
      \rput[b]{0}(0,0.7){$N_0$}
      \psline[linestyle=dotted](5,0)(5,-0.7)
      \rput[t]{0}(5,-0.72){$t_0$}
      \psline[linestyle=dotted](31,0)(31,-0.7)
      \rput[t]{0}(31,-0.72){$t_1$}
      \psline[linestyle=dotted](7.5,0.5)(7.5,-0.7)
      \rput[t]{0}(7.5,-0.72){$t'$}
      \psline[linestyle=dotted](12,0.7)(12,-0.7)
      \rput[t]{0}(12,-0.72){$n_0$}
    \end{pspicture}
    \caption{Case $e_0\geq 0\wedge \lambda\leq 1$}
    \label{fig:maxiter:solve:case2}
  }\hfill
  \parbox[t]{0.3\linewidth}{
    \vspace*{1ex}

    \psset{xunit=0.33em,yunit=4em,plotpoints=50,Dx=100,Dy=100}
    \begin{pspicture}(-10,-1)(40,1.3)
      \psaxes(0,0)(-10,0)(40,0)
      \psplot{-10}{40}{x dup 5 sub 14 mul sin exch 1.02 exch exp mul
        0.5 mul}
      \psplot[linestyle=dashed]{-10}{40}{1.02 x exp 0.5 mul}
      \psline[linestyle=dashed](-10,0.5)(40,0.5)
      \rput[b]{0}(-7,0.52){$e_0$}
      \psline[linestyle=dotted](-4,0.7)(-4,0)
      \rput[b]{0}(-4,0.7){$N_0$}
      \psline[linestyle=dotted](0,0.7)(0,0)
      \rput[l]{0}(-1,0.8){$t_\mathit{hull}$}
      \psline[linestyle=dotted](5,0)(5,-0.7)
      \rput[t]{0}(5,-0.72){$t_0$}
      \psline[linestyle=dotted](31,0)(31,-0.7)
      \rput[t]{0}(31,-0.72){$t_1$}
      \psline[linestyle=dotted](9.5,0.7)(9.5,-0.7)
      \rput[t]{0}(9.5,-0.72){$t'$}
      \psline[linestyle=dotted](17,0.7)(17,-0.7)
      \rput[t]{0}(17,-0.72){$n_0$}
    \end{pspicture}
    \caption{Case $e_0\geq 0\wedge \lambda>1$}
    \label{fig:maxiter:solve:case3}
  }\null
\end{figure*}

\subsubsection{Bounding single coefficients}
\label{sec:maxiter:complex:single}

We consider here the case
$\phi(n)=\mu\varphi[\lambda,\theta,r,k=0](n)$ where
$\varphi[\lambda,\theta,r,k](t)$ is
defined by Eqn.~(\ref{eq:bound:varphi}). This case
$k=0$, which can be generalized
to
$$
\phi(t) = \alpha e^{\gamma t}\sin(\theta t+\varphi)
$$
with $\alpha>0$, $\theta\in]0,\pi[$ and $\varphi\in[-\tfrac{3\pi}{2},\tfrac{\pi}{2}]$.
We have
$$
\phi'(t)=\alpha  e^{\gamma t}\bigl(\gamma\sin(\theta t+\varphi)+\theta\cos(\theta t+\varphi)\bigr)
$$

We look for the first $t_0\geq N_0$ such that $f(t_0)=0$ and
$f'(t_0)>0$. This implies that $\theta t_0+\varphi=2\pi k_0\geq
\theta N_0+\varphi$. Hence we define
\begin{align*}
  k_0 & = \vceil{\frac{\theta N_0 + \varphi}{2\pi}} &
  t_0 &= \frac{2\pi k_0-\varphi}{\theta}
\end{align*}
\begin{enumerate}
\item Case $e_0\leq 0$ (see Fig.~\ref{fig:maxiter:solve:case1}):
  \begin{itemize}
  \item If $\phi(N_0)>e_0$ then the solution is $N_0$.
  \item Otherwise, we solve $\phi(t)=e_0/\alpha$ with the
    Newton-Raphson method, starting from $t_0$, and we obtain a
    solution
    $t'\in[t_0-\frac{1}{\theta}\tfrac{\pi}{2},t_0]$ which
    necessarily satisfies $t'\geq N_0$. We then define
    $n_0 =
    \begin{cases}
      t'+1 &\text{ if } t'\in\mathbb{N} \\
      \ceil{t'} &\text{ otherwise}
    \end{cases}
    $
    which is necessarily the solution, because
    $\theta<\pi$ and $n_0-t_1<1$, hence $\theta(n_0-t_1)<\pi$.
  \end{itemize}
\item Case $e_0>0 \wedge \lambda\leq 1$ (see Fig.~\ref{fig:maxiter:solve:case2}):
  As we need successive tries, we will iterate on an integer $i$
  from $0$ to some bound. For each iteration, we define
  \begin{align*}
    k_i &= k_0+i &
    t_i &= t_0+2\pi i=\frac{2\pi k_i-\varphi}{\theta}
  \end{align*}
  \begin{itemize}
  \item If $\phi(t_i+\frac{1}{\theta}\tfrac{\pi}{2})\leq e_0$,
    then the solution is $\min \emptyset=+\infty$.
  \item Otherwise we solve $\phi(t)=e_0/\alpha$ with the
    Newton-Raphson method, starting from $t_i$, and we obtain a
    solution
    $t'\in[t_i,t_i+\frac{1}{\theta}\tfrac{\pi}{2}]$. We then consider
    $n_0 = \begin{cases}
      t'+1 &\text{ if } t'\in\mathbb{N} \\
      \ceil{t'} &\text{ otherwise}
    \end{cases}
    $ that satisfies $n_0-t_1<\frac{\pi}{\theta}$. If
    $\phi(n_0)>e_0$ then it is the solution, otherwise we try the
    next $i$ up to some arbitrary bound.
  \end{itemize}
\item Case $e_0>0 \wedge \lambda> 1$ (see Fig.~\ref{fig:maxiter:solve:case3}): \\
  As we need successive tries, we again iterate on an integer $i$
  from $0$ to some bound.  We proceed however slightly differently
  for this last case. We first define
  \begin{align*}
    t_\mathit{hull}&=\frac{\log (e_0/\alpha)}{\gamma}
  \end{align*}
  which is the rank from which the expanding sinusoid
  may be greater than $e_0$.  We now look for the first $t_0\geq
  N_0$, $t_0\geq t_\mathit{hull}$ such that $f(t_0)=0$ and
  $f'(t_0)>0$. This implies that $\theta t_0+\varphi=2\pi k_0\geq
  \theta \max(N_0,t_\mathit{hull})+\varphi$. Adding the
  iteration $i$, we define
  \begin{align*}
    k_i & = \vceil{\frac{\theta \max(N_0,t_\mathit{hull}) + \varphi}{2\pi}} + i &
    t_i &= \frac{2\pi k_i-\varphi}{\theta}
  \end{align*}
  We solve $\phi(t)=e_0$ starting from $t_i$ and we obtain a
  solution
  $t'\in[t_i,t_i+\frac{1}{\theta}\tfrac{\pi}{2}]$. We then consider
  $n_0 = \begin{cases}
    t'+1 &\text{ if } t'\in\mathbb{N} \\
    \ceil{t'} &\text{ otherwise}
  \end{cases}
  $ that satisfies $n_0-t_1<\frac{\pi}{\theta}$. If
  $\phi(n_0)>e_0$ then it is the solution, otherwise we try the
  next $i$ up to some arbitrary bound.
\end{enumerate}

\subsubsection{Bounding some linear combination of two coefficients}
\label{sec:maxiter:complex:two}

We handle only the case $\gamma_1=\gamma_2$, $\theta_1=\theta_2$,
$k_1=k_2=0$ and $r_1\neq r_2$. In this case we exploit
Proposition~\ref{prop:bound:complex:two:associated} to reduce the
problem to the one considered in Section~\ref{sec:maxiter:complex:single}.

\end{document}